\documentclass[conference]{IEEEtran}
\IEEEoverridecommandlockouts

\usepackage{amsmath}
\usepackage{amsfonts}
\usepackage{float} 
\usepackage{srcltx}
\usepackage{mathrsfs} 
\usepackage{multirow}

\usepackage{graphicx}
\usepackage{epsfig}
\usepackage{psfrag}
\usepackage{subfigure}

\usepackage{setspace}

\usepackage[]{units} 
\usepackage{url} 
\usepackage[dvips]{color}
\usepackage{verbatim} 

\newcommand{\cardinality}[1]{\ensuremath{\lvert#1\rvert}}

\newcommand{\vol}[1]{\ensuremath{\textrm{Vol}\left( #1 \right)}}

\newcommand{\mspcb}{\ensuremath{\:}}
\newcommand{\mspcc}{\ensuremath{\;}}
\newcommand{\mspcd}{\ensuremath{\quad}}

\newcommand{\mset}[1]{\ensuremath{\mathcal{#1}}}
\newcommand{\mvec}[1]{\ensuremath{\textbf{#1}}}

\newcommand{\mDefine}{\ensuremath{ \stackrel{\triangle}{=} }}

\newcommand{\mGoesTo}{\ensuremath{\rightarrow}}
\newcommand{\mGoesToAs}[1]{\ensuremath{\underset{#1}{\longrightarrow}}}

\newcommand{\mRR}{\ensuremath{\mathbb{R}}}

\newcommand{\mZZ}{\ensuremath{\mathbb{Z}}}

\newcommand{\mDistributedBy}{\ensuremath{\sim}}

\newtheorem{theorem}{Theorem}

\newtheorem{proposition}{Proposition}

\newcommand{\ErrexpMAC}{\ensuremath{E^\textrm{MAC}}} 
\newcommand{\ErrexpSU}{\ensuremath{E^\textrm{SU}}} 
\newcommand{\ErrexpLinearSU}{\ensuremath{E^\textrm{SU}_\textrm{L}}} 
\newcommand{\ErrexpLinearBestKnownSU}{\ensuremath{\underline{E}^\textrm{SU}_\textrm{L}}} 
\newcommand{\EachievableSU}{\ensuremath{\underline{E}^\textrm{SU}}} 
\newcommand{\ErSU}{\ensuremath{E_r^\textrm{SU}}} 
\newcommand{\EexSU}{\ensuremath{E_{ex}^\textrm{SU}}} 
\newcommand{\RcrSU}{\ensuremath{R_{cr}^\textrm{SU}}} 
\newcommand{\RexSU}{\ensuremath{R_{ex}^\textrm{SU}}} 
\newcommand{\ErG}{\ensuremath{E_r^\textrm{G}}} 
\newcommand{\ErGc}{\ensuremath{E_{r3}^\textrm{G}}} 
\newcommand{\ErSWc}{\ensuremath{E_{r3}^\textrm{SW}}} 
\newcommand{\ErUG}{\ensuremath{\overline{E_r^\textrm{G}}}} 
\newcommand{\Rstruct}{\ensuremath{\mset{R}_\textrm{struct}}} 
\newcommand{\DistributedNestingErrexp}{\ensuremath{E_r^\textrm{struct}}} 

\begin{document}

\allowdisplaybreaks

\title{Distributed Structure: Joint Expurgation for the Multiple-Access Channel}

\author{
\authorblockN{Eli Haim}
\authorblockA{Dept. of EE-Systems,
TAU\\
Tel Aviv, Israel \\
Email: elih@eng.tau.ac.il}
\and
\authorblockN{Yuval Kochman}
\authorblockA{School of CSE,
HUJI\\
Jerusalem, Israel \\
Email: yuvalko@cs.huji.ac.il}
\and
\authorblockN{Uri Erez\authorrefmark{1}}
\authorblockA{Dept. of EE-Systems,
TAU\\
Tel Aviv, Israel \\
Email: uri@eng.tau.ac.il}
\thanks{$^*$ This work was supported in part by the U.S. - Israel Binational Science
Foundation under grant 2008/455.}
\thanks{The results of this paper were presented in part in the International Symposium on Information Theory, 2011, St. Petersburg, Russia. Another part will be presented in the International Symposium on Information Theory, 2012, Cambridge, MA.}
}

\maketitle


\begin{abstract}
In this work we show how an improved lower bound to the error exponent of the memoryless multiple-access (MAC) channel is attained via the use of linear codes, thus demonstrating that structure can be beneficial even in cases where there is no capacity gain.
We show that if the MAC channel is modulo-additive, then any error probability, and hence any error exponent, achievable by a linear code for the corresponding single-user channel, is also achievable for the MAC channel.
Specifically, for an alphabet of prime cardinality, where linear codes achieve the best known exponents in the single-user setting and the optimal exponent above the critical rate, this performance carries over to the MAC setting.
At least at low rates, where expurgation is needed, our approach strictly improves performance over previous results, where expurgation was used at most for one of the users.
Even when the MAC channel is not additive, it may be transformed into such a channel. While the transformation is lossy, we show that the distributed structure gain in some ``nearly additive'' cases outweighs the loss, and thus the error exponent can improve upon the best known error exponent for these cases as well.
Finally we apply a similar approach to the Gaussian MAC channel. We obtain an improvement over the best known achievable exponent, given by Gallager, for certain rate pairs, using lattice codes which satisfy a nesting condition.
\end{abstract}

%
\section{Introduction}
%
\label{sec:intro}

The error exponent of the multiple access (MAC) channel is a long-standing open problem. While superposition and successive decoding methods lead to capacity, they may not be optimal in the sense of error probability: the decoding process may be improved by considering that the transmission of other users is a codeword, rather than noise. However, finding the optimal performance is a difficult task, beyond the difficulties encountered in a point-to-point channel.
Early results include the works of Slepian and Wolf~\cite{SlepianWolf73MAC}, Gallager~\cite{Gallager:1985} and Pokorny and Wallmeier~\cite{PokornyWallmeier1985}. Applying the results of \cite{SlepianWolf73MAC} to the important special case of a (modulo) additive MAC channel, e.g., the binary symmetric case, it follows that the random-coding exponent of the corresponding single-user channel is achievable for the MAC channel. This exponent is optimal above the critical rate \cite{SlepianWolf73MAC}. However, for lower rates it is outperformed by the expurgated exponent (in the single-user case). The reason that the expurgated exponent is not achieved in~\cite{SlepianWolf73MAC} is that the sum of two good (expurgated) single-user codebooks does not result in a good single-user one, and in particular, the sum of two codebooks with good minimum-distance properties may not be good in that respect.
Liu and Hughes~\cite{LiuHughes:1996} and recently Nazari et al.~\cite{NazariAnastasopoulosPradhan:2010:arxiv} have proposed improvements over earlier results. Specifically, Nazari et al. suggest to apply expurgation to one of the codebooks. While this certainly improves performance, it still does not allow to achieve the single-user expurgated exponent.

For additive MAC channels we make the basic observation, that by ``splitting'' a linear codebook between the users, any error probability achievable in the corresponding single-user channel using linear codebooks is achievable for the MAC channel as well. This implies for prime (e.g. binary) alphabets, that the best currently known error exponents for \emph{any} code (not necessarily linear) are achievable for the MAC channel, including the random-coding and expurgated exponents. The improvement over previous results stems from the use of linear codes, which are inherently expurgated; thus using them provides ``joint expurgation'' even in a distributed setting.

But what happens outside the special case of additive channels? To see this, we go back to settings where the application of linear codes to additive communications networks has a capacity advantage, see e.g.~\cite{KrithivasanPradhan:2009, WilsonNarayananPfisterSprintson:2010, PhilosofZamirErezKhisti:2011}.
We are inspired by the fact that in the context of first-order (capacity) analysis of networks, the advantage of linear codes has indeed been extended to some non-additive channels~\cite{ErezZamir:2008}.
In~\cite{ErezZamir:2008} a modulo-lattice transformation is derived, that allows to obtain a virtual additive channel from any original MAC channel, albeit with a loss of capacity. It is shown in~\cite{ErezZamir:2008} that in some situations, the gain offered by the ability to use linear codes outweighs the loss inflicted by the transformation. In this work we adopt the same ideas to the MAC exponent problem: we show that for MAC channels that are ``nearly symmetric'', indeed the transformation in conjunction with using linear codes improves upon the best known exponents so far at low rates. We note that when one considers less symmetric channels, the results of~\cite{NazariAnastasopoulosPradhan:2010:arxiv} outperform those of the new scheme.

The technique we propose, of splitting a linear codebook, may be interpreted as nested linear codebooks, where the codebook of one user is nested in that of the the other. We leverage this observation to extend our approach beyond discrete alphabets, and consider the exponent of the Gaussian MAC channel. As in the discrete case, the sum of the codebooks, as seen by the decoder, is a single linear code, which is inherently expurgated. However, unlike the discrete case, the exponent we obtain is inferior to the single-user exponent. Moreover, there is a rate loss in comparison to the single user capacity. Still, despite this loss, we improve upon the best previously known error exponent~\cite{Gallager:1985} for certain power pairs and certain rate pairs.

The rest of the paper is organized as follows.
We start with the discrete MAC channel, where
Section~\ref{sec:discrete:channel_model} presents the background and definitions.
Then, Section~\ref{sec:discrete:coding_technique} describes the coding technique for the modulo additive MAC channel.
Section~\ref{sec:discrete:transformation} describes a technique for transforming a general discrete MAC channel into a modulo additive one (with some loss).
Section~\ref{sec:binary_case} presents an analysis of the special case of the binary MAC channel.
We then turn to the Gaussian MAC channel, where after some background on error exponents of Gaussian channels in Section~\ref{sec:Gaussian:preliminaries}, 
Section~\ref{sec:Gaussian:distributive_nesting_code} describes the coding technique by using distributive nesting, and derives its performance.
Finally, in Section~\ref{sec:discussion} we discuss the results and give some conclusions.

%
\section{Discrete Case: Background and Definitions}
%
\label{sec:discrete:channel_model}

\subsection{Single-User Channel}
\label{sec:discrete:channel_model:signel_user}

Consider the single-user discrete memoryless channel (DMC) defined by $P_{Y|X} (\cdot|\cdot)$,
where $X$ and $Y$ are the channel input and output, respectively, with discrete alphabets $\mset{X}$ and $\mset{Y}$.
We recall some results regarding the error exponent of this channel, see \cite{GallagerBook1968}.

The error exponent of the channel is defined as
\begin{align}
\label{eq:exponent}
\ErrexpSU(R) = \limsup_{n \mGoesTo \infty} -\frac{1}{n} \log \epsilon_n,
\end{align}
where $\epsilon_n$ is the minimal possible error probability of codes (averaged over the codewords) with block length $n$ and rate~$R$.

The best known achievable error exponent for this channel, denoted by $\EachievableSU(R)$, is given by the maximum between the expurgated error exponent $\EexSU(R)$ and the random-coding error exponent $\ErSU(R)$, where~\cite{GallagerBook1968}:
\begin{align*}
\ErSU(R) &= \max_{0 \leq \rho \leq 1} \max_{P_X} \left[ E_0(\rho, P_X) - \rho R \right],
\end{align*}
where $P_X$ is some distribution over the scalar channel input~$\mset{X}$ and
\begin{align*}
E_0(\rho, P_X) \mDefine
- \log \sum_{y \in \mset{Y}} \left(
\sum_{x \in \mset{X}} P_X(x) P_{Y|X}(y|x)^{1/(1+\rho)},
\right)^{1+\rho}.
\end{align*}
The expurgated exponent is given by:
\begin{align*}
\EexSU(R) &= \sup_{\rho \geq 1} \max_{P_X} \left[ E_x(\rho, P_X) - \rho R \right],
\end{align*}
where
\begin{align*}
E_x(\rho, P_X) \mDefine
- \rho \log
\sum_{x_1 \in \mset{X}}
\sum_{x_2 \in \mset{X}}
P_X(x_1) P_X(x_2)
\\ \times
\left(
\sum_{y \in \mset{Y}}
\sqrt{P_{Y|X}(y|x_1) P_{Y|X}(y|x_2)}
\right)^{1/\rho}.
\end{align*}

The expurgated exponent is larger than the random-coding exponent below some rate $\RexSU$ (this range is thus called ``the expurgation region''). Above the critical rate $\RcrSU$, the random-coding exponent is larger, and is known to be optimal.

\subsection{MAC Channel}
\label{sec:discrete:channel_model:MAC}

Consider a two-user discrete memoryless MAC channel $P_{Y|X_1,X_2}$, where $X_1, X_2$ are the channel inputs and $Y$ is its output, over (discrete) alphabets $\mset{X}_1, \mset{X}_2$ and $\mset{Y}$ respectively.
Denote the codebook of user $i$ by $\mset{C}_i$, and its rate by $R_i = \nicefrac{1}{n}\log \cardinality{\mset{C}_i}$.

Following Slepian and Wolf~\cite{SlepianWolf73MAC}, we define the error event as the event that at least one of the messages from the message pair is decoded in error.\footnote{Other definitions, leading to an error exponent region, were considered in~\cite{WengPradhanAnastasopoulos:2008}.}
The error exponent of the MAC channel is defined as
\begin{align}
\label{eq:MAC:errexp:def}
\ErrexpMAC( R_1, R_2) = \limsup_{n \mGoesTo \infty} -\frac{1}{n} \log \epsilon_n,
\end{align}
where $\epsilon_n$ is the minimal possible error probability for codes of length $n$, with the rate-pair
$(R_1,R_2)$.

Slepian and Wolf~\cite{SlepianWolf73MAC} found an achievable error exponent that is given by the minimum of three random-coding error exponents corresponding to different error events.\footnote{Slepian and Wolf~\cite{SlepianWolf73MAC} considered a more general case of a MAC channel with correlated sources, and obtained with an achievable error exponent that is the minimum of four error exponents. Gallager~\cite{Gallager:1985} reformulated this result to the channel-only problem, in which case, the results simplify to only three of the exponents.} The first two correspond to making an erroneous decision on one message, by a genie-aided decoder, i.e., one that has knowledge of the message of the other user as side information. The third error event corresponds to making an erroneous decision in both messages. For positive rates, the third exponent, denoted by $\ErSWc$, is equal to the error exponent of a single-user channel with input equal to the input-pairs of the MAC channel (still statistically independent symbol-pairs) and with rate equal to the sum rate. Therefore $\ErSWc$ depends only on the sum rate (see also~\cite{Gallager:1985}).
Each of these three events amounts to an error event over a single-user channel. Therefore, each exponent is equal to Gallager's random coding error exponent~\cite{GallagerBook1968} for the corresponding single-user channel.

\subsection{Additive-Noise Single-User Channel}
\label{sec:discrete:channel_model:additive:single_user}

Consider the following DMC:
\begin{align}
\label{eq:single-user}
Y = X \oplus N,
\end{align}
where all variables are defined over the alphabet  $\mathbb{Z}_m = \{0,1,\ldots,m-1\}$ and $\oplus$ denotes addition over this alphabet, i.e., modulo an integer $m$.
The noise $N$ is additive, i.e., statistically independent of the channel input $X$.

The random-coding error exponent of this channel~\eqref{eq:single-user} can be expressed in terms of R\'{e}nyi entropy (see e.g.~\cite{ErezZamir2001}):
\begin{align*}
\ErSU(R) = \max_{0 \leq \rho \leq 1} \rho \left[ \log m - h_\frac{1}{1+\rho}(N) - R \right],
\end{align*}
where $h_\beta(N)$ is the R\'{e}nyi entropy of order $\beta$, and is defined by\footnote{In the limit of $\rho=1$, R\'{e}nyi entropy becomes the Shannon entropy.}
\begin{align}
\label{eq:Reny_entropy}
h_\beta(N) = \frac{\beta}{1-\beta} \log \left( \sum_{n=0}^{m-1} P_N(n)^\beta \right)^\frac{1}{\beta}.
\end{align}
As for general channels, the best known error exponent of this channel is larger than $\ErSU(R)$ in the expurgation region, as can be seen in Figure~\ref{fig:single_user_ee_comparison}.
\begin{figure}[tb]
\centering
\includegraphics[width=\columnwidth]{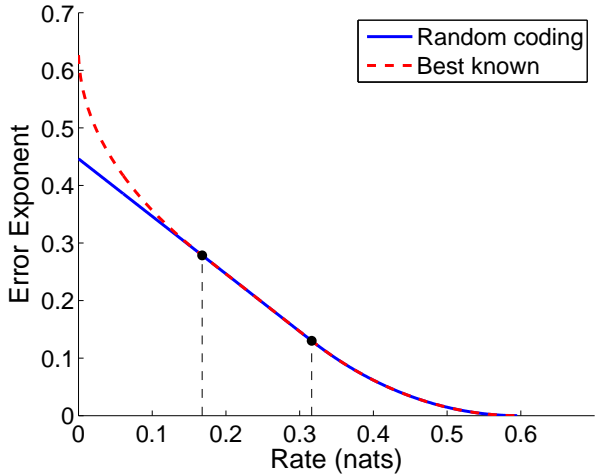} 
\caption{Comparing the random coding error exponent of a additive-noise single-user channel, with the best known error exponent. The channel is an additive binary symmetric channel (BSC) with noise $\mDistributedBy\textrm{Bernoulli}(0.02)$.
The two dots show the expurgation rate and the critical rate of this channel respectively.}
\label{fig:single_user_ee_comparison}
\end{figure}


In the context of additive channels, it is important to consider \emph{linear} codes. We define a linear code~$\mset{C}$ via a $k \times n$ generating matrix $G$, by
\begin{align}
\label{eq:linear_code}
\mset{C} = \{ \mvec{c}: \mvec{c} = \mvec{u} G, \mspcc \mvec{u}\in\mathbb{Z}_m^k \},
\end{align}
The rate\footnote{We assume a full-rank matrix $G$.} is equal to $R = \nicefrac{k}{n} \log m$.
Clearly, for any rate, there exists a linear code of this rate asymptotically as $n \mGoesTo \infty$.
We define  $\ErrexpLinearSU(R)$ to be the error exponent of linear codes, i.e., as \eqref{eq:exponent}, except that $\epsilon_n$ is the minimal possible error probability of \emph{linear} codes only.
We also denote the best known achievable error exponent of linear codes by $\ErrexpLinearBestKnownSU(R)$.
We note that for single-user additive channels of the form \eqref{eq:single-user}, when the alphabet size $m$ is a prime, the best known error exponent of linear codes, $\ErrexpLinearBestKnownSU(R)$, is equal to the best known error exponent of the channel $\EachievableSU(R)$ (see \cite{Dobrushin63, GallagerBook1968, BargForney2002}), and in particular is optimal above the critical rate.
In addition, we note that for linear codes the average error probability (over the codewords) is equal to the maximal error probability, due to their structure (i.e., the maximum likelihood decoding regions are identical up to translation).

\subsection{Additive-Noise MAC Channel}
\label{sec:discrete:channel_model:additive_MAC}

A channel which is of particular interest in this work is the additive MAC channel
\begin{align}
\label{eq:mac:additive_channel}
Y &= X_1 \oplus X_2 \oplus N,
\end{align}
where all variables are defined over the alphabet  $\mathbb{Z}_m = \{0,1,\ldots,m-1\}$ and $\oplus$ denotes addition over this alphabet, i.e., modulo an integer $m$.
The noise $N$ is additive, i.e., is statistically independent of the pair $(X_1,X_2)$.

Viewing the joint codebook $X = X_1 \oplus X_2$ as a single-user codebook, we get the channel~\eqref{eq:single-user} (over the same alphabet)
\begin{align}
\label{eq:single-user:associated}
Y = X \oplus N,
\end{align}
which we call the \emph{associated single-user channel} of~\eqref{eq:mac:additive_channel}.
Any codebook pair $(\mset{C}_1,\mset{C}_2)$ for the MAC channel can be used to construct a \emph{corresponding codebook} for its associated single-user channel, by the Minkowski sum codebook $\mset{C} = \mset{C}_1 + \mset{C}_2$.
However, not every single-user codebook $\mset{C}$ can be decomposed in such a manner.
Moreover, since the associated single-user channel is equivalent to cooperation between the encoders, then when comparing the MAC channel to its associated single-user one, it follows that
\begin{align*}
\ErrexpMAC( R_1, R_2) \leq \ErrexpSU(R_1+R_2).
\end{align*}

For additive MAC channels, in~\cite{SlepianWolf73MAC} it is shown that
\footnote{Since it can be shown that for the discrete additive MAC channel, out of the three error exponents discussed above, the third, $\ErSWc$, always dominates.}
\begin{align*}
\ErrexpMAC(R_1, R_2) \geq \ErSU(R_1+R_2).
\end{align*}
Thus, $\ErrexpMAC(R_1, R_2)$ is equal to $\EachievableSU(R_1+R_2)$ above the expurgation rate and optimal above the critical rate.
However, the best known error exponent for the associated single-user channel though, is larger in the expurgation region; recall Figure~\ref{fig:single_user_ee_comparison}.

We note that simple time sharing, where every user uses an expurgated codebook, improves on the Slepian-Wolf random-coding bound~\cite{SlepianWolf73MAC} in some cases, particularly for low enough rates and as the channel noise becomes weaker.

Since~\cite{SlepianWolf73MAC}, there were several improvements~\cite{LiuHughes:1996, NazariAnastasopoulosPradhan:2010:arxiv} to the achievable error exponent. However, these do not close the gap to the best known error exponent of the associated single-user channel.
In the next section, we close this gap for modulo-additive MAC channels, by attaining expurgation for all users.

%
\section{Coding for Modulo-Additive Discrete MAC Channels}
%
\label{sec:discrete:coding_technique}


In this section we first describe a coding scheme for additive-noise discrete MAC channels, that achieves the best known error exponent of linear codes for its associated single-user channel. In particular, for prime alphabet size, it achieves the best known error exponent for the associated single-user channel. This is equivalent to full cooperation of the encoders, and thus it is optimal (in terms of error exponent) whenever the optimum is known for the single-user channel (i.e., above its critical rate).


Consider the additive-noise MAC channel, as given in~\eqref{eq:mac:additive_channel}, with alphabet size~$m$.
We construct a codebook pair for the MAC channel using linear codes.
We use a good linear code for the associated single-user channel \eqref{eq:single-user:associated}, which we decompose into two linear sub-codes, one for each user.


Let $G$ be a $k \times n$ generating matrix of a linear code $\mset{C}$~(see~\eqref{eq:linear_code}) with rate~$R = \nicefrac{k}{n} \log m$.
For some integers $k_1+k_2=k$, define the rates \[R_i = \frac{k_i}{n} \log m , \ i=1,2. \]
Decompose the codeword $\mvec{c}$ into two codewords:
\begin{align}
\mvec{c} &= \mvec{c}_1 \oplus \mvec{c}_2
= \left( \mvec{u}_{1 \times k_1} \mid \mvec{u}_{1 \times k_2} \right) \left(
\begin{array}{c}
G_{k_1 \times n}
\\ \hline
G_{k_2 \times n}
\end{array}
\right)
\\&\mDefine \mvec{u}_1 G_1 \oplus \mvec{u}_2 G_2.
\end{align}
Thus, we have a pair of codebooks:
\begin{align}
\mset{C}_i = \{ \mvec{c}_i: \mvec{c}_i = \mvec{u}_i G_i, \mspcc \mvec{u}_i\in\mathbb{Z}_m^{k_i} \}; \mspcd i=1,2.
\end{align}
Therefore, the sum of codewords is indistinguishable from a codeword of the single-user code with $R = R_1 + R_2$.
Clearly, for any rate-pair such a construction is possible asymptotically as $n \mGoesTo \infty$. A similar claim holds for a general number of users as well. We thus have the following.


\begin{proposition}
The coding technique above achieves the best error probability of linear codes for the single-user channel. Thus, it achieves the exponent $\ErrexpLinearBestKnownSU(R)$, and for prime $m$ it achieves $\EachievableSU(R)$ as well. 
\end{proposition}

\emph{Remark}: The result also holds for an additive MAC channel~\eqref{eq:mac:additive_channel} where the alphabet size~$m$ is a power of a prime and addition is over the field.


The previously best known error exponent is given by Nazari~et~al.~\cite{NazariAnastasopoulosPradhan:2010:arxiv}. In their derivation, codewords are expurgated from only one of the codebooks. Since our bound achieves the best known error exponent of the associated single-user channel (for the special case of additive MAC channels with prime alphabet size), it must be at least as good the one found by Nazari et al.
For low enough rate-pairs we expect our bound to be strictly better, since full expurgation is required in order to achieve the error exponent of the associated single-user channel.
When considering more than two users, the gap is expected to increase since expurgation of one user becomes less significant. In the sequel, we show how this advantage can be leveraged to non-additive MAC channels.  


The distributed-structure code construction presented in this section can be interpreted in terms of nested linear codes. Two linear codes are nested if one of them (the \emph{coarse} codebook) is a subset of the other (the \emph{fine} codebook). For the code described in this section, the single-user codebook is the fine code $\mset{C}$. The coarse codebook $\mset{C}_1 \subseteq \mset{C}$ is the codebook of the first user. This forms a quotient group $\mset{C}/\mset{C}_1$, where any member of this group (i.e., coset) is a different ``translate'' of the coarse codebook $\mset{C}_1$. A selection of representatives from every coset forms a codebook $\mset{C}_2$ for the second user. Any such selection of coset representatives leads to the same fine code $\mset{C}_1 \oplus \mset{C}_2 = \mset{C}$, and therefore is a good selection. As a special case, in the code construction which is described above, the coset representatives are selected such that they form a linear code.


This nested linear codes approach can be extended to the continuous alphabet case.
Consider the modulo-additive channel where the channel alphabets and noise are continuous:
\begin{align}
\label{eq:continuous:modulo_additive_channel:one_dimentional}
Y = (X_1 + X_2 + N) \bmod 1.
\end{align}
First assume that the input alphabets are $p^{-1} \cdot \mZZ$ (which is equal to the integers multiplied by $\nicefrac{1}{p}$), where $p$ is prime.
Linear codes achieve the best known error exponent for a single user modulo-prime additive channel (See Section~\ref{sec:discrete:coding_technique}).
Therefore for prime $p$, nested linear codes achieve this exponent.
Taking $p$ to infinity one can approach as closely as desired an optimal codebook pair for continuous alphabet inputs.
This will lead to a distributed coding technique for the Gaussian MAC channel in the sequel.

%
\section{Transforming a General Discrete MAC Channel into an Additive Channel}
%
\label{sec:discrete:transformation}

With the aim of applying a similar scheme to general (non-additive) discrete memoryless MAC channels $P_{Y|X_1,X_2}$, in this section we describe a method for transforming such channels into additive-noise MAC channels. We refer to the obtained channel after the transformation as the resulting \emph{virtual} channel.
The transformation is a discrete and scalar modification of the \emph{Modulo-Lattice Transformation} for continuous MAC channels~\cite{ErezZamir:2008}.

The transformation is defined for any finite alphabet size $m$, regardless of the alphabet sizes of the inputs and the output. For simplicity, we assume throughout this section that $m$ is prime.
Let $V_i \in \mathbb{Z}_m$ be the input of the $i$th user to the virtual channel, and $U_i \mDistributedBy \textrm{Uniform}(\mathbb{Z}_m)$ be its dither (i.e., common randomness at the $i$th encoder and at the decoder), which is statistically independent of the dither of the other user and of $V_1,V_2$.
Each encoder computes $X'_i = V_i \oplus U_i$ and applies a scalar precoding function $f_i:\mathbb{Z}_m \mGoesTo \mset{X}$ to it. The inputs to the channel are therefore given by
\begin{align}
X_i = f_i(X'_i).
\end{align}
Note that due to the dither, $X'_i$ is uniformly distributed over $\mathbb{Z}_m$ and is statistically independent of $V_1,V_2$.
Let
\begin{align*}
S = k_1 X'_1 \oplus k_2 X'_2,
\end{align*}
where $k_i \in \mZZ_m$, and multiplication is over $\mZZ_m$. Let $\hat{S} = g(Y)$ be some scalar ``estimator'' function of~$S$ from the channel output $Y$.
Denote the estimation error by $N = \hat{S} \ominus S$ (i.e., a subtraction operation over $\mZZ_m$).
We define the output of the virtual channel as
\begin{align}
Y' &\mDefine \hat{S} \ominus (k_1 U_1 \oplus k_2 U_2)
\end{align}

\begin{proposition}[The virtual MAC channel]
\label{prop:virtual}
Applying the transformation leads to the following virtual channel:
\begin{align}
\label{eq:virtual_MAC}
Y' = k_1 V_1 \oplus k_2 V_2 \oplus N,
\end{align}
where $N=\hat{S} \ominus S$ is statistically independent of the channel inputs $(V_1,V_2)$.
\end{proposition}

\begin{proof}
We have:
\begin{align}
Y' &= \hat{S} \ominus (k_1 U_1 \oplus k_2 U_2) \nonumber
\\&= \hat{S} \ominus S \oplus S \ominus (k_1 U_1 \oplus k_2 U_2) \nonumber
\\&= N \oplus k_1(V_1 \oplus U_1) \oplus k_2(V_2 \oplus U_2) \ominus (k_1 U_1 \oplus k_2 U_2)\nonumber
\\&= k_1 V_1 \oplus k_2 V_2 \oplus N. \nonumber
\end{align}
\end{proof}


Notice that the transformation is not unique, and one is free to choose the alphabet size $m$, the precoding functions $f_i(\cdot)$ and the estimator of $S$.
We call any virtual MAC channel~\eqref{eq:virtual_MAC} that can be obtained by some choice of parameters, a feasible virtual MAC channel.
Since we assume that the alphabet size $m$ is prime, it follows that a feasible single-user channel: 
\begin{align}
\label{eq:virtual_MAC:associated_single_user}
Y = X \oplus N,
\end{align}
is the associated single-user channel of a feasible virtual MAC channel \eqref{eq:virtual_MAC}.


Applying this transformation to any MAC channel, we have the following.
\begin{proposition}
\label{prop:virtual_MAC:bound}
Let $\epsilon_n$ be the best error probability achievable with a code of length $n$ on a MAC channel. Then \[ \epsilon_n \leq \tilde \epsilon_n, \] where $\tilde \epsilon_n$ is the best error probability achievable by a \emph{linear} code of the same length on a feasible virtual MAC channel~\eqref{eq:virtual_MAC}.
\end{proposition}


Applying Proposition~\ref{prop:virtual_MAC:bound} to exponents, leads to our main result:
\begin{theorem}
\label{cor:MAC_achieves_the_vitual_SU}
For any MAC channel, and any associated feasible single-user channel~\eqref{eq:virtual_MAC:associated_single_user} with alphabet of prime cardinality,
\begin{align*}
\ErrexpMAC( R_1, R_2) \geq \EachievableSU(R_1+R_2).
\end{align*}
\end{theorem}

\emph{Remarks:}
\begin{itemize}


\item Notice that this transformation is lossy in terms of capacity. However, since the resulting channel is an additive-noise channel, efficient coding techniques and known bounds can be easily applied.
In particular, for MAC channels, expurgation in all the users can be applied by using linear codes as in Section~\ref{sec:discrete:coding_technique}.


\item We expect the benefit from this coding technique to outweigh the loss when the channel is ``close'' to additive. In the next section, we give a binary example which illustrates this property with a single parameter.


\item We note that this transformation is applicable to various non-additive network problems, where structure can improve the best-known achievable \emph{rate region} (see e.g.~\cite{KrithivasanPradhan:2009, WilsonNarayananPfisterSprintson:2010, PhilosofZamirErezKhisti:2011}). In such settings, the gain will appear also as a ``capacity gain'' rather than only in the probability of error.
\end{itemize}

%
\section{Binary Case}
%
\label{sec:binary_case}
In this section we confine the discussion to binary MAC channels, i.e. channels with binary inputs and output.
We denote this general (i.e., non-additive) channel $P_{Y|X_1,X_2}$ as:
\begin{align}
\label{eq:MAC_channel:discrete:additive_representation}
Y = X_1 \oplus X_2 \oplus Z,
\end{align}
where the ``additive'' noise $Z \mDefine Y \oplus (X_1 \oplus X_2)$ \emph{may depend} on the channel input pair $(X_1,X_2)$.

\subsection{Analysis of the Virtual Channel}
A natural choice for the parameter $m$ of the transformation is clearly $m=2$.
We select $f(x)=x$, $k_1=k_2=1$ and $\hat{S}=g(Y)=Y$. This leads to the following effective noise of the virtual channel: $N = Z$. However, $Z$ is statistically independent of $(V_1,V_2)$ due to the transformation. Therefore, the probability distribution of the effective additive noise $N$ of the virtual channel, is equal to the marginal distribution of $Z$, i.e.:
\begin{align}
\label{eq:binary_case:N}
N \mDistributedBy \textrm{Bernoulli}(\gamma),
\end{align}
with
\begin{align}
\label{eq:binary_case:p}
\gamma = \frac{1}{4} \sum_{x_1, x_2 \in \{0,1\}} \Pr(Z=1|X_1=x_1,X_2=x_2).
\end{align}
The virtual channel is then an additive MAC channel given by:
\begin{align}
Y = V_1 \oplus V_2 \oplus N,
\end{align}
where $N$, given in~\eqref{eq:binary_case:N}-\eqref{eq:binary_case:p}, is statistically independent of $(V_1,V_2)$.

\subsection{Example: Almost Additive Binary MAC Channel}

We now use the analysis of the previous subsection in order to study an example of an almost additive-noise binary MAC channel.
Specifically, we consider the following MAC channel:
\begin{align}
Y &= X_1 \oplus X_2 \oplus Z \label{eq:binary_MAC_example:channel}
\\Z &\mDefine Z_1 \oplus 1_{\{X_1 \neq X_2\}} \cdot Z_2, \label{eq:binary_MAC_example:noise}
\end{align}
where $Z_1 \mDistributedBy \textrm{Bernoulli}(q)$, $Z_2 \mDistributedBy \textrm{Bernoulli}(p)$, and $1_{\{X_1 \neq X_2\}}$ is the indicator function of the event $X_1 \neq X_2$.
The value of $p$ determines the deviation of the channel from additivity. For small $p$ the channel is nearly an additive MAC channel.

In Figure~\ref{fig:distance_from_additive} we compare the resulting error exponent of the virtual channel with the Slepian-Wolf random coding exponent~\cite{SlepianWolf73MAC} for symmetric rate pairs.\footnote{For the channel parameter $q=0.1$ of Fig.~\ref{fig:distance_from_additive} and zero rate-pair, the exponent of time sharing between expurgated codebooks is below the Slepian-Wolf random coding exponent for all $p$.} For the comparison we take the limit of zero rate-pair, where the gain due to expurgation is maximal.
\begin{figure}[tb]
\centering
\includegraphics[width=\columnwidth]{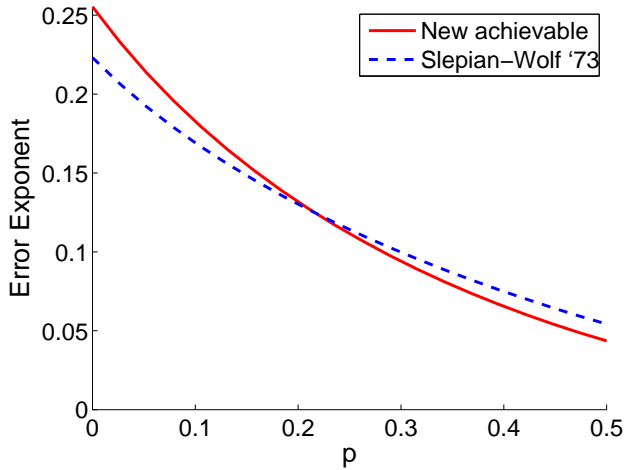} 
\caption{Comparing error exponents for the almost additive binary MAC~\eqref{eq:binary_MAC_example:channel}-\eqref{eq:binary_MAC_example:noise}. The dashed line is the Slepian-Wolf bound~\cite{SlepianWolf73MAC}. The solid line is the error exponent of the virtual channel, which is achieved according to Corollary~\ref{cor:MAC_achieves_the_vitual_SU}. Here $q=0.1$ and the comparison is at zero rate-pair.}
\label{fig:distance_from_additive}
\end{figure}
%
As $p$ increases, the coding technique developed in this paper gains less since the channel transformation looses more as the channel becomes less additive.

For small enough $p$ and for low enough rates we expect this bound to be strictly larger then the best known error exponent for this channel~\cite{NazariAnastasopoulosPradhan:2010:arxiv}. This is since~\cite{NazariAnastasopoulosPradhan:2010:arxiv} applies expurgation only to the user with larger rate, while the bound presented here achieves two-user expurgation.
Nazari et al.~\cite{NazariAnastasopoulosPradhan:2010:arxiv} studied a non-symmetric example, where $\Pr(Z=1|X_1=1,X_2=1) = \frac{1}{2}$, and all the other conditional probabilities of $Z$ are equal to $0.01$. In this case the coding scheme described here is inferior to the one of~\cite{NazariAnastasopoulosPradhan:2010:arxiv}, as expected since the channel is far from being additive.

%
\section{Gaussian Channels: Preliminaries}
%
\label{sec:Gaussian:preliminaries}

In this section we recall some results for Gaussian channels and give some definitions, similar to the ones presented in Section~\ref{sec:discrete:channel_model} for the discrete case.
There are two differences in the model with respect to the discrete case: the channel input is continuous and is subject to a power constraint; the additive noise is restricted to Gaussian.

\subsection{Single-User Channel}
\label{sec:Gaussian:preliminaries:single_user}

Consider the single-user additive white Gaussian noise (AWGN) channel:
\begin{align}
\label{eq:single_user:channel}
Y = X + Z,
\end{align}
where $X \in \mRR$ is the input to the channel and is subject to a power constraint:
\begin{align}
\label{eq:single_user:power_constarints}
\frac{1}{n} \sum_{j=1}^n x_j^2 \leq P,
\end{align}
where $n$ is the codeword length.
The noise $Z \mDistributedBy \mathcal{N}(0,N)$ is an additive noise, i.e., is statistically independent of $X$.
The capacity of this channel is given by $C(A) \mDefine \nicefrac{1}{2} \log(1+A)$, where $A\mDefine\nicefrac{P}{N}$ is the signal to noise ratio (SNR).

The error exponents $\ErrexpSU(R,A)$, $\EachievableSU(R,A)$, $\EexSU(R,A)$ and $\ErSU(R,A)$  are defined similar to the ones in Section~\ref{sec:discrete:channel_model:signel_user}.
However, while in the unconstrained case, the optimal distribution (in the sense of error exponent) of the input symbols was given by an i.i.d. distribution, in the constrained case, the optimal distribution includes statistical dependence due to codebook shaping (see e.g.,~\cite[Chapter~7]{GallagerBook1968}). Specifically, spherical shell codebooks are used, and this leads to the best known error exponents.
In particular, above the critical rate, the random coding error exponent is optimal, and is given by
\begin{align*}
&\ErSU(R,A) = \frac{A}{4\beta}\left[
(\beta+1) - (\beta-1) \sqrt{ 1 + \frac{4\beta}{A(\beta-1)} }
\right]
\\&\phantom{\ErSU(R,} +\frac{1}{2} \log \left\{ \beta - \frac{A(\beta - 1)}{2}\left[ \sqrt{1+\frac{4\beta}{A(\beta-1)}} - 1 \right] \right\},
\end{align*}
where $\beta \mDefine \exp(2R)$. The critical rate is
\begin{align*}
\RcrSU(A) = \frac{1}{2} \log \gamma,
\end{align*}
where
\begin{align*}
\gamma = \frac{1}{2}\left( 1 + \frac{A}{2} + \sqrt{ 1 + \frac{A^2}{4} } \right).
\end{align*}
Below the critical rate, the random coding error exponent is given by
\begin{align*}
\ErSU(R,A) = 1 - \gamma + \frac{A}{2} + \frac{1}{2} \log \left( \gamma - \frac{A}{2} \right) + \frac{1}{2}\log \gamma - R.
\end{align*}
The expurgated exponent is given by
\begin{align*}
\EexSU(R,A) = \frac{A}{4}\left( 1 - \sqrt{1 - \exp(-2R)} \right)
\end{align*}
and it is larger than the random coding error exponent below the expurgated rate, given by:
\begin{align*}
\RexSU(A) = \frac{1}{2} \log \left( \gamma - \frac{A}{4} \right).
\end{align*}

In the high-SNR limit $A \gg 1$, $\EachievableSU(R,A)$ approaches the Poltyrev exponent, defined by:
\begin{align}
\label{eq:poltyrev_exponent}
E_P(\mu) \mDefine \left\{
\begin{array}{ll}
0, & \mu \leq 1
\\
\frac{1}{2}\left[ (\mu-1) - \log \mu \right], & 1 < \mu \leq 2
\\
\frac{1}{2} \log \frac{e \mu}{4} , & 2 \leq \mu \leq 4
\\
\frac{\mu}{8} , & \mu \geq 4
\end{array}
\right. ,
\end{align}
where $\mu = (1+A)\exp(-2R)$.
The Poltyrev exponent at the critical values $\mu=1$, $\mu=2$ and $\mu=4$ correspond to $C(A)$, $\RcrSU(A)$ and $\RexSU(A)$, respectively.

\subsection{MAC Channel}
\label{sec:Gaussian:preliminaries:MAC}

Consider a two-user memoryless Gaussian MAC channel:
\begin{align}
\label{eq:awgn_mac:channel}
Y &= X_1 + X_2 + Z,
\end{align}
where $X_1, X_2$ are the inputs to the channel, and are subject to power constraints as in~\eqref{eq:single_user:power_constarints} with powers $P_1$ and $P_2$ respectively.
The additive noise $Z \mDistributedBy \mathcal{N}(0,N)$ is independent of the pair $(X_1,X_2)$.
Without loss of generality, let $P_1\geq P_2$.

We denote the SNRs by $A_i \mDefine P_i/N$ $(i=1,2)$.
Denote the codebook of user $i$ by $\mset{C}_i$, and its rate by $R_i = \nicefrac{1}{n}\log \cardinality{\mset{C}_i}$.
The error exponent of the channel $\ErrexpMAC( R_1, R_2, A_1, A_2 )$ is defined similar to~\eqref{eq:MAC:errexp:def}.

The best known achievable error exponent for this channel is given by Gallager~\cite{Gallager:1985}, denoted here by $\ErG(R_1,R_2,A_1,A_2)$.
This exponent is derived via analyzing the same error events which are defined by Slepian and Wolf~\cite{SlepianWolf73MAC}, and described in Section~\ref{sec:discrete:channel_model:MAC}. Thus, this exponent is also given by the minimum between three exponents.
Inspired by the fact that spherical-shell codebooks are good for the single-user Gaussian channel, Gallager derived a lower (achievable) bound $\ErG(R_1,R_2,A_1,A_2)$ using such codebooks for the MAC channel.
In the sequel, we use an upper bound on the exponent $\ErG(R_1,R_2,A_1,A_2)$ derived in~\cite{Gallager:1985}, which is equal to the exponent which corresponds to the third error event, $\ErGc(R_1+R_2,A_1,A_2)$ (which is analogous to the $\ErSWc(R_1+R_2)$ of the discrete case):
\begin{align}
\nonumber
\ErUG(R,A_1,A_2) &\mDefine \ErGc(R,A_1,A_2)
\\\nonumber
&= (1+\rho) \log \frac{e\sqrt{\theta_1 \theta_2}}{1+\rho}
- \frac{\theta_1 + \theta_2}{2}
\\&\phantom{=} + \frac{\rho}{2} \log \left( 1 + \frac{A_1}{\theta_1} + \frac{A_2}{\theta_2} \right) - \rho R,
\label{eq:spherical_shell:UbOnEr3}
\end{align}
where $R \mDefine R_1+R_2$, and the expression is optimized over $\rho \in [0,1]$ and over $r_i$ for:
\begin{align}
\theta_i \mDefine (1+\rho)(1 - 2 r_i P_i), \mspcd \theta_i \in [0,1+\rho].
\end{align}

The associated single-user channel of~\eqref{eq:awgn_mac:channel} is defined as in the discrete case in Section~\ref{sec:discrete:channel_model:additive_MAC}:
\begin{align}
\label{eq:Gaussian:associated_single_user}
Y = X + Z,
\end{align}
where $X$ represents $X_1 + X_2$.
However, in contrast to the discrete case, here there are power constraints on $X_1$ and $X_2$. Since there is no cooperation between the transmitters, the best power we can hope to achieve in the associated single-user channel is the sum of powers. Therefore the power constraint of input of the associated single-user channel is given by $P \mDefine P_1+P_2$.
%
%
Thus, the error-exponent of the associated single-user channel
\begin{align*}
\ErrexpSU(R_1+R_2,A_1+A_2)
\end{align*}
serves as a benchmark for the error exponent of the MAC channel.
%
The exponent $\ErG(R_1,R_2,A_1,A_2)$ is strictly smaller than the single-user exponent. Moreover, it is strictly smaller than the \emph{random coding} error exponent $\ErSU(R_1+R_2,A_1+A_2)$~\cite{Gallager:1985}.

\subsection{Lattices Preliminaries}
\label{sec:Gaussian:lattice_preliminaries}

This section presents mathematical background that is required for the code construction in the continuous alphabets case and its error-probability analysis. For a more thorough treatment of lattices, the reader may refer to~\cite{ErezLitsynZamir:2005} and the references therein.

A lattice $\Lambda$ is a discrete subgroup of the Euclidean space $\mRR^n$ with the ordinary vector addition operation.
A Lattice may be specified in terms of a generating matrix. Thus, an $n \times n$ real-valued matrix $G$ defines a lattice by
\begin{align}
\Lambda = \left\{ \lambda = G \mvec{x}: \mvec{x} \in \mZZ^n \right\}.
\end{align}

A coset of $\Lambda$ is any translate of it, i.e., $\mvec{x}+\Lambda$, where $\mvec{x} \in \mRR^n$.
Any set $\Omega$ of coset representatives is called a \emph{fundamental region}.
Therefore, every $\mvec{x} \in \mRR^n$ can be uniquely expressed as $\mvec{x} = \lambda + \mvec{r}$, where $\lambda \in \Lambda, \mvec{r} \in \Omega$.
The Voronoi region of $\Lambda$ with respect to the origin, denoted by $\mathcal{V}$, is a set of minimum Euclidean norm coset representatives of $\Lambda$, where ties are broken arbitrarily.
The nearest-neighbor quantizer of a point $\mvec{x} \in \mRR^n$ is defined by:
\begin{align}
\label{eq:quantizer}
\begin{array}{cc} 
Q(\mvec{x}) = \lambda, & \textrm{if } \mvec{x} \in \lambda + \mathcal{V} \textrm{ and } \lambda \in \Lambda.
\end{array}
\end{align}

\subsubsection{Good Lattices}
\label{sec:good_lattices}
We define two notions of ``good'' lattices.
Let $r_\Lambda^\textrm{cov}$ denote the covering radius of $\Lambda$, i.e., the radius of the smallest ball containing the Voronoi region $\mathcal{V}$.
Let $r_\Lambda^\textrm{effec}$ denote the effective radius of the Voronoi region, i.e., the radius of a sphere having the same volume as $\mathcal{V}$.
We say that a sequence of lattices $\Lambda^{(n)} \in \mRR^n, \mspcb n=1,2,\ldots,$ is \emph{good for covering} if $\liminf_{n \mGoesTo \infty} r_{\Lambda^{(n)}}^\textrm{cov}/r_{\Lambda^{(n)}}^\textrm{effec} = 1$.
Choosing a sequence of lattices which are good for covering with $r_{\Lambda^{(n)}}^\textrm{cov} = \sqrt{nP}$ leads to $r_{\Lambda^{(n)}}^\textrm{effec} = \sqrt{n(P - \delta_n)}$, where $\lim_{n\mGoesTo\infty} \delta_n = 0^+$, and in addition  that $\nicefrac{1}{n}\sum_{i=1}^n x_i^2 \leq P$ for any point in $\mvec{x} \in \mathcal{V}$.
Rogers~\cite{Rogers59} established the existence of sequence of lattices which are good for covering (which we denote as \emph{Rogers-good}).

A sequence of lattices $\Lambda^{(n)} \in \mRR^n \mspcb, n=1,2,\ldots,$ is said to be \emph{Poltyrev-good} if for an $n$-dimensional vector $\mvec{Z}$ with i.i.d. Gaussian entries with zero mean and power $N$:
\begin{align}
\label{eq:Poltyrev_good_condition}
\Pr\left(  \mvec{Z} \notin \mathcal{V}^{(n)} \right) < e^{-n\left[E_P\left(\rho_{\Lambda^{(n)}}^2\right) - o_n(1)\right]},
\end{align}
where $\mathcal{V}^{(n)}$ is the Voronoi region of $\Lambda^{(n)}$, 
$\rho_{\Lambda^{(n)}}$ is the \emph{Voronoi-to-noise} ratio:
\begin{align}
\label{eq:voronoi_to_noise_ratio:n}
\rho_{\Lambda^{(n)}} \mDefine \frac{r_{\Lambda^{(n)}}^\textrm{effec}}{\sqrt{n N}}
= \frac{[\vol{\mathcal{V}^{(n)}}]^{1/n}}{\sqrt{2 \pi e N}} + o_n(1),
\end{align}
with $o_n(1) \mGoesTo 0$ as $n \mGoesTo \infty$, and $E_P(\mu)$ is the Poltyrev exponent defined in~\eqref{eq:poltyrev_exponent}.
The Poltyrev exponent is the best known achievable error exponent in the \emph{unrestricted} additive noise white Gaussian (AWGN) setting~\cite{Poltyrev:94:AWGN}, where the rate is measured per unit volume.

\subsubsection{Nested Lattice Codes}
\label{sec:nested_lattice_codes}
Here we recall nested-lattice codes for a single-user AWGN channel. The construction of codes for the Gaussian MAC channel is described in Section~\ref{sec:Gaussian:distributive_nesting_code}, and builds on nested lattice codes.

We say that a coarse lattice $\Lambda_0$ is nested in a fine lattice $\Lambda_1$ if $\Lambda_0 \subseteq \Lambda_1$, i.e., $\Lambda_0$ is a sublattice of $\Lambda_1$.
We denote their Voronoi regions with respect to the origin by $\mathcal{V}_0$ and $\mathcal{V}_1$ respectively, and the volumes of the Voronoi regions by $V_0$ and $V_1$ respectively.

A translate of $\Lambda_0$ by a point of $\Lambda_1$ is called a coset of $\Lambda_0$ relative to $\Lambda_1$.
Any set of coset representatives of $\Lambda_0$ relative to $\Lambda_1$ is called a \emph{nested-lattice code} $(\Lambda_1,\Lambda_0)$.
However, in power constrained codebooks, we select a set of coset representatives with minimal power, which may be defined by taking the intersection of a fine lattice $\Lambda_1$ with the Voronoi region (w.r.t. the origin) of a sublattice $\Lambda_0$, i.e., $\mset{C} = \Lambda_1 \cap \mathcal{V}_0$. This special case of nested lattice codes is denoted by $(\Lambda_1,\Lambda_0)_\textrm{Vor}$.

Thus, the number of codewords is equal to $V_0/V_1$. We call $(V_0/V_1)^{1/n}$ the nesting ratio of the lattices. The code rate is thus equal to the normalized per dimension logarithm of the nesting ratio: $R = \nicefrac{1}{n} \log (V_0/V_1)$.


The existence of simultaneously Rogers-good and Poltyrev-good nested lattices sequence $\Lambda_0 \subseteq \Lambda_1 \subseteq \cdots \subseteq \Lambda_{L-1}$ for any nesting level $L$ and any choice of nesting ratios is shown in~\cite{KrithivasanPradhan:2007:online}.

%
\section{Coding for Gaussian MAC Channels: Distributed Nesting}
%
\label{sec:Gaussian:distributive_nesting_code}
In this section we describe a code construction for the MAC channel, which is based on distributed structure. It consists of two codebooks, where both are subsets of the same lattice, thus their Minkowski sum forms a subset of the fine lattice. It has the property that every pair of codewords results in a different point of the fine lattice. In addition, the fine lattice is a good lattice for the associated single-user channel.
Hence, inherently from the code, the decoding is done jointly for the two users.

The resulting exponent from this coding scheme is given by the following.  
\begin{theorem} \label{thm:Gaussian} For a Gaussian MAC channel with SNRs $(A_1,A_2)$ and rates $(R_1,R_2)$,
\begin{align}
\label{eq:Gaussian:MAC:structured_code:error_exponent}
E^\textrm{MAC} (R_1,R_2,A_1,A_2) & \geq \DistributedNestingErrexp (R_1,R_2,A_1,A_2) \nonumber \\
& \mDefine E_p(\min(\mu_1,\mu_2)), \end{align} where the Poltyrev exponent $E_P(\cdot)$ was defined in \eqref{eq:poltyrev_exponent}, and
\begin{align}
\label{eq:mu1}
\mu_1=\mu_1(R_1,R_2,A_1,A_2) &= A_1 \exp[-2(R_1+R_2)], \\ \label{eq:mu2}
\mu_2=\mu_2(R_1,R_2,A_1,A_2) &= A_2 \exp[-2R_2].
\end{align}
\end{theorem}

Unlike the discrete modulo-additive case, this exponent is strictly smaller than the exponent of the associated single-user channel. Furthermore, it is positive only inside the region
\begin{align}
\label{eq:rate_region}
\Rstruct \mDefine \left\{ (R_1,R_2) : \mspcb
R_1+R_2 \leq \frac{1}{2} \log A_1; \mspcb
R_2 \leq \frac{1}{2} \log A_2
\right\},
\end{align}
which is strictly smaller than the MAC capacity region (see Figure~\ref{fig:Rstruct_region}). Thus, the code presented here is sub-optimal in terms of capacity. Still, for certain rate-pairs it improves on the best previously known error exponent for the Gaussian MAC channel as derived in~\cite{Gallager:1985}.

\begin{figure*}[htb]
\centering
\psfrag{&D1}[rc]{{\footnotesize $\frac{1}{2} \log A_1$}}
\psfrag{&C1}[ll]{             {\footnotesize $\frac{1}{2} \log (1+A_1)$}}
\psfrag{&D12}[cc]{{\footnotesize $\frac{1}{2} \log \frac{A_1}{A_2}$}}
\psfrag{&D2}[cc]{{\footnotesize $\frac{1}{2} \log A_2$}}
\psfrag{&C2}[cc]{{\footnotesize $\frac{1}{2} \log (1+A_2)$}}
\psfrag{&mu1}[cc]{{\small $E_P(\mu_1)$}}
\psfrag{&mu2}[cc]{{\small $E_P(\mu_2)$}}
\includegraphics[scale=0.5]{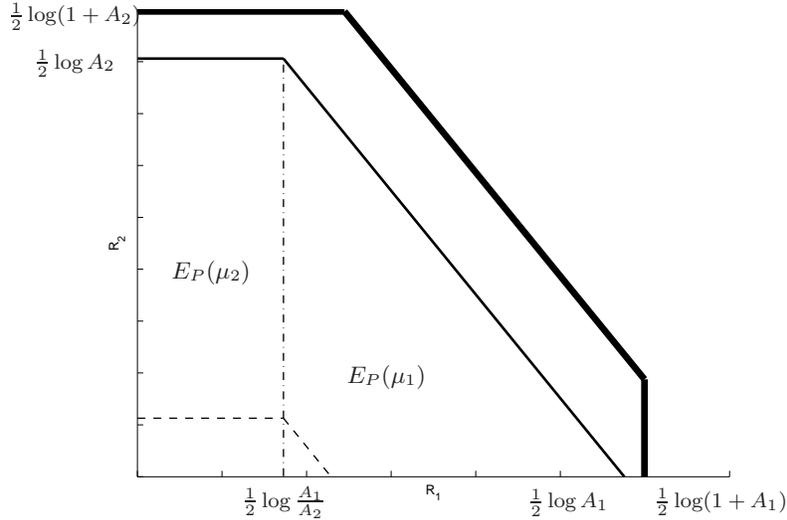}
\caption{The boundary of the capacity region of the Gaussian MAC is given the thick solid line. The boundary of the achievable region of the structured code, $\Rstruct$, is given by the thin solid line. The dash-dot line is the point where $\mu_1=\mu_2$, which is equivalent to $R_1 = \nicefrac{1}{2} \log \nicefrac{P_1}{P_2}$. This line separates between the region which $\mu_1$ is dominant and the region which $\mu_2$ is dominant in the error exponent of the structured code~\eqref{eq:Gaussian:MAC:structured_code:error_exponent}. The dash line denotes the boundary of the expurgation region of this exponent, which is $\min\{\mu_1,\mu_2\} \geq 4$.}
\label{fig:Rstruct_region}
\end{figure*}

In Section~\ref{sec:modulo_lattice_MAC_channel} we present an extension of the one-dimensional continuous modulo additive channel~\eqref{eq:continuous:modulo_additive_channel:one_dimentional} from Section~\ref{sec:discrete:coding_technique} to higher dimensions, in which the exponent of the associated single-user channel is indeed achieved. Then in Section~\ref{sec:one_dimensional_motivation} we discuss the difference between this problem and the power-constrained Gaussian channel, in terms of coding and performance. Section~\ref{sec:codebook_construction} describes the coding scheme, while Section~\ref{sec:error_analysis} evaluates the resulting error exponent, thus proving Theorem~\ref{thm:Gaussian}. Finally, in Section~\ref{sec:comparison} we compare the performance to previously-known bounds. 

\subsection{Distributed Structure for the Modulo Lattice Additive MAC Channel with Continuous Alphabets}
\label{sec:modulo_lattice_MAC_channel}

Extending the one-dimensional continuous modulo additive channel~\eqref{eq:continuous:modulo_additive_channel:one_dimentional} to higher dimensions, we get the following channel:
\begin{align}
\label{eq:modulo_lattice_MAC_channel}
\mvec{Y} = (\mvec{X}_1 + \mvec{X}_2 + \mvec{Z}) \bmod \Lambda_0,
\end{align}
where $\Lambda$ is a lattice over $\mRR^n$ and $\mvec{Z}$ is a Gaussian i.i.d. vector with zero mean and power $N$.
The best known error exponent of the associated single-user channel
\begin{align}
\mvec{Y} = (\mvec{X} + \mvec{Z}) \bmod \Lambda_0
\end{align}
is achieved by a nested lattice codebook $(\Lambda_1, \Lambda_0)$, with some lattice $\Lambda_0 \subseteq \Lambda_1$~(see~\cite{ErezZamir:2004}).
This exponent is also achievable in the MAC channel by using a nested pair lattices (in which~$\Lambda_0$ is nested), since they form an indistinguishable codebook from the one of the associated single user.
Hence, extending the coding technique of Section~\ref{sec:discrete:coding_technique} to the continuous case also achieves the best known error exponent (the Poltyrev exponent~\eqref{eq:poltyrev_exponent} with $\mu=\rho_\Lambda^2$, where the Voronoi-to-noise $\rho_\Lambda=\nicefrac{r_{\Lambda}^\textrm{effec}}{\sqrt{n N}}$ as defined in~\eqref{eq:voronoi_to_noise_ratio:n}) of the associated of the single-user channel, and is optimal above its critical rate.

\subsection{The Effect of the Power Constraint}
\label{sec:one_dimensional_motivation}


Comparing the Gaussian MAC channel~\eqref{eq:awgn_mac:channel} to the modulo-lattice MAC~\eqref{eq:modulo_lattice_MAC_channel}, there are two differences: the first is that the Gaussian channel does not perform a modulo lattice operation and the second is that the powers of the channel inputs are constrained.
In order to demonstrate the effect of the power constraint, in this section we consider the case of scalar (one-dimensional lattice) codebooks, i.e., uncoded transmission.


For the associated single-user channel~\eqref{eq:Gaussian:associated_single_user} of the Gaussian MAC channel, at high SNR, the optimal codebook for uncoded transmission approaches a pulse-amplitude modulation (PAM) constellation, which is a set of equidistant points, symmetrically around zero, that satisfy the power constraint. This constellation is a one-dimensional nested-lattice codebook $(\mZZ,L\mZZ)_\textrm{Vor}$, where $L$ is the constellation size and it is odd. Therefore, $\mset{C} = \left\{ 0, 1, \ldots, L-1 \right\} - (L-1)/2 = \left\{ 0, \pm 1, \ldots, \pm(L-1)/2 \right\}$.
%
%
This demonstrates one effect of the power constraint: while in the associated single-user channel of the modulo-lattice MAC channel~\eqref{eq:modulo_lattice_MAC_channel} the coset representatives could be selected arbitrarily, in the power-constrained case, the coset representatives must be selected such that the the power constraint is satisfied. Thus we select minimal Euclidean norm coset representatives. This effect carries over to the multidimensional construction described in Section~\ref{sec:codebook_construction}.


In order to build a good scalar codebook pair for the Gaussian MAC channel, we use the observation that under certain conditions, the Minkowski sum of two PAM signals is also a PAM signal. Therefore a good (scalar) code for the associated single-user channel is constructed in a distributed manner.
Specifically, consider the triplet of nested lattices $L_0 \mZZ \subseteq L_1 \mZZ \subseteq \mZZ$, where the ratio $L_0/L_1$ is odd. Then the codebook of the first user is given by the nested-lattice code $(L_1\mZZ,L_0\mZZ)_\textrm{Vor}$, i.e. $\mset{C}_1 = \left\{ 0, L_1, 2L_1, \ldots, L_0-L_1 \right\}-(L_0-L_1)/2$. The codebook of the second user is given by the nested-lattice code $(\mZZ,L_1\mZZ)_\textrm{Vor}$, i.e. $\mset{C}_2 = \left\{ 0, 1, \ldots, L_1-1 \right\} - (L_1-1)/2$. Thus the Minkowski sum of the codebooks is equal to the nested-lattice code $(L_2\mZZ,\mZZ)_\textrm{Vor}$, i.e., $\mset{C} = \mset{C}_1 + \mset{C}_2 = \left\{ 0, \pm 1, \ldots, \pm(L_0-1)/2 \right\}$.
The transformation $c = c_1+c_2$ where $c_1 \in \mset{C}_1, c_2 \in \mset{C}_2$ is injective, i.e., the codeword pair can be resolved from their sum without ambiguity (see Figure~\ref{fig:one_dim_example}).
\begin{figure}[tb]
\centering
\subfigure[The codebook of the second user (with the weaker power) $\mset{C}_2$.]{
\includegraphics[scale=0.6, trim=10mm 0mm 10mm 110mm, clip=true]{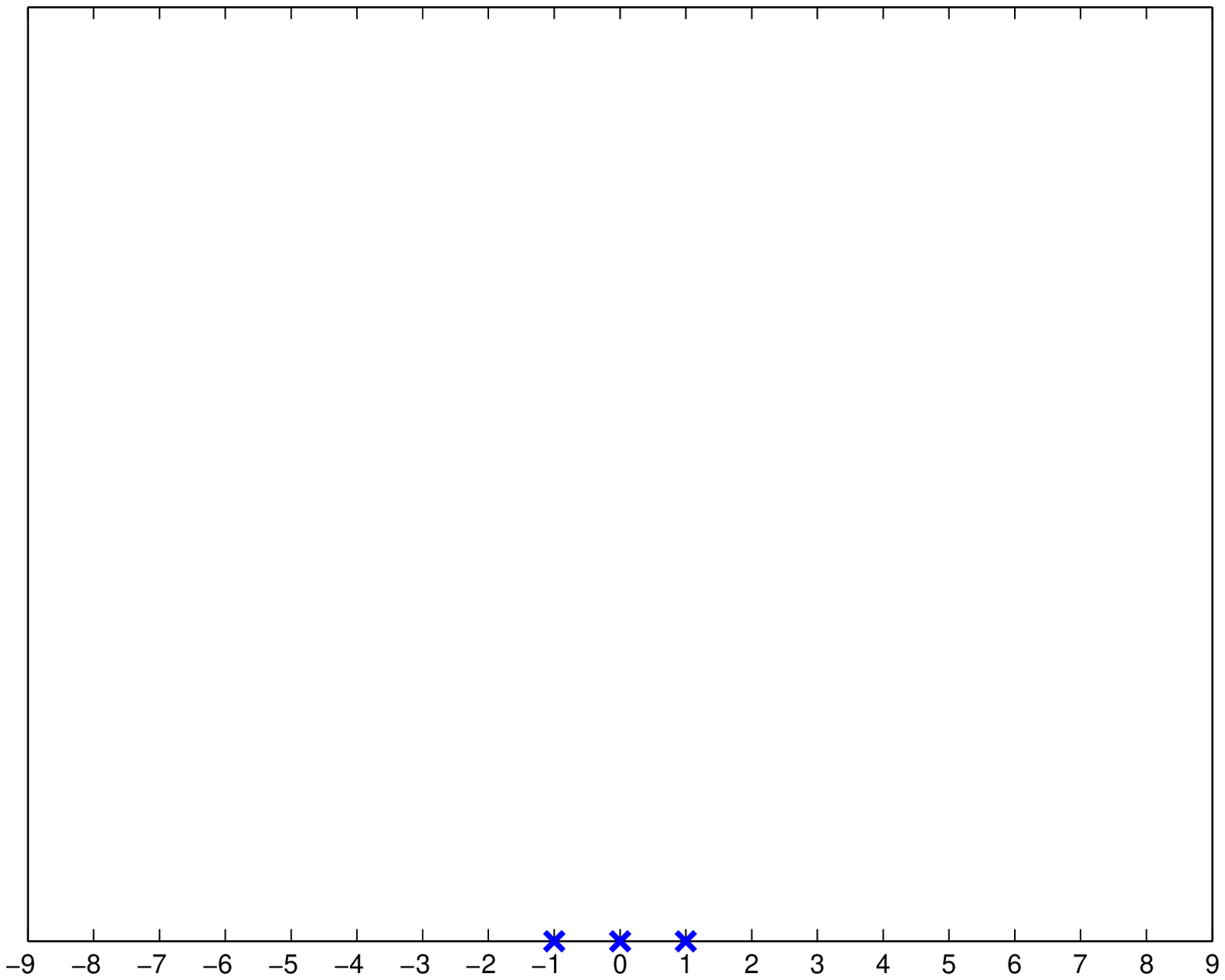}
}
\subfigure[The codebook of the first user (with the stronger power) $\mset{C}_1$.]{
\includegraphics[scale=0.6, trim=10mm 0mm 10mm 110mm, clip=true]{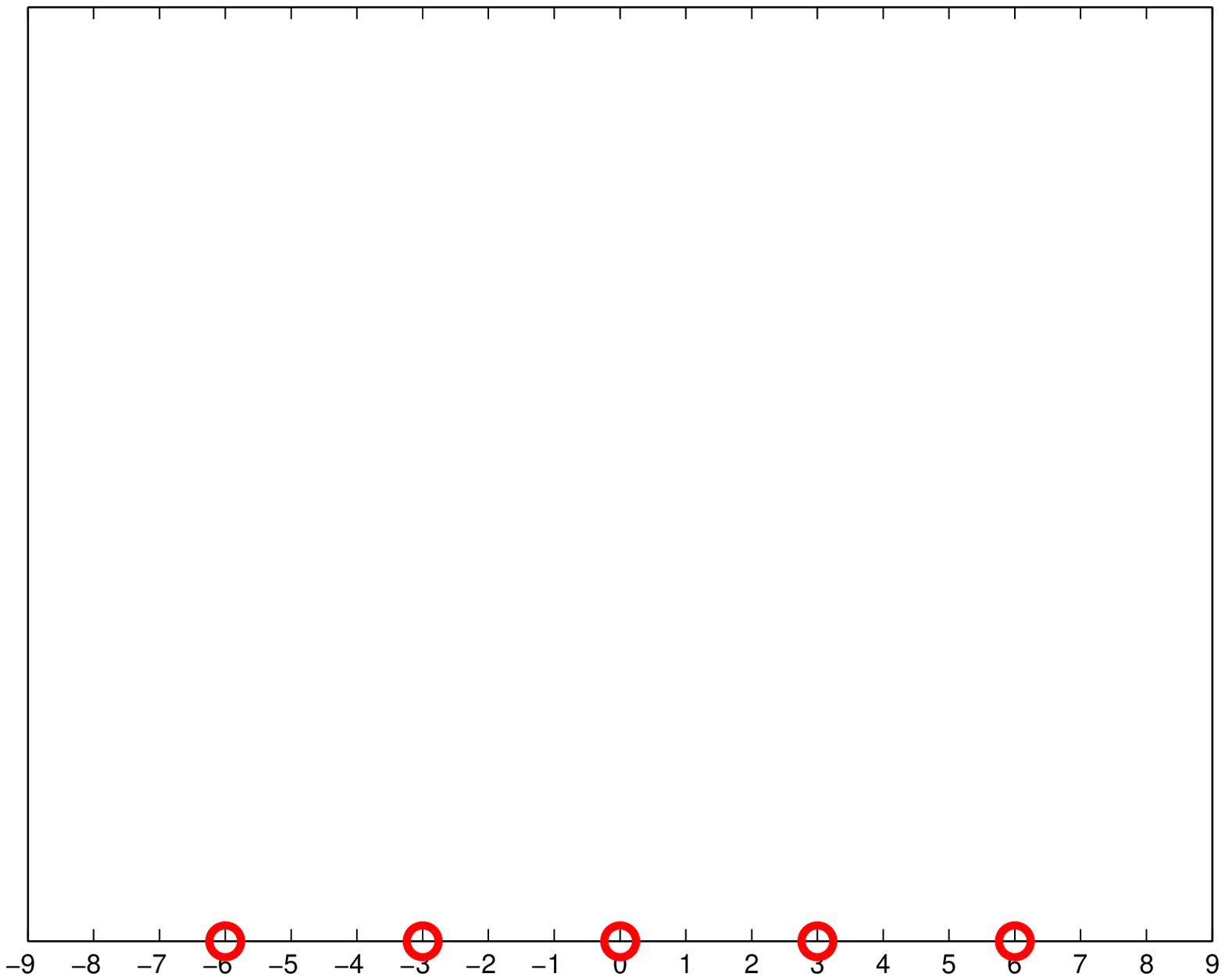}
}
\subfigure[The joint codebook as seen at the decoder, i.e., the Minkowski sum $\mset{C}_1+\mset{C}_2$.]{
\psfrag{&-6+C2}[ll]{{\scriptsize $-6+C_2$}}
\psfrag{&-3+C2}[ll]{{\scriptsize $-3+C_2$}}
\psfrag{&+0+C2}[ll]{{\scriptsize $\phantom{+}0+C_2$}}
\psfrag{&+3+C2}[ll]{{\scriptsize $+3+C_2$}}
\psfrag{&+6+C2}[ll]{{\scriptsize $+6+C_2$}}
\includegraphics[scale=0.6, trim=10mm 0mm 10mm 82mm, clip=true]{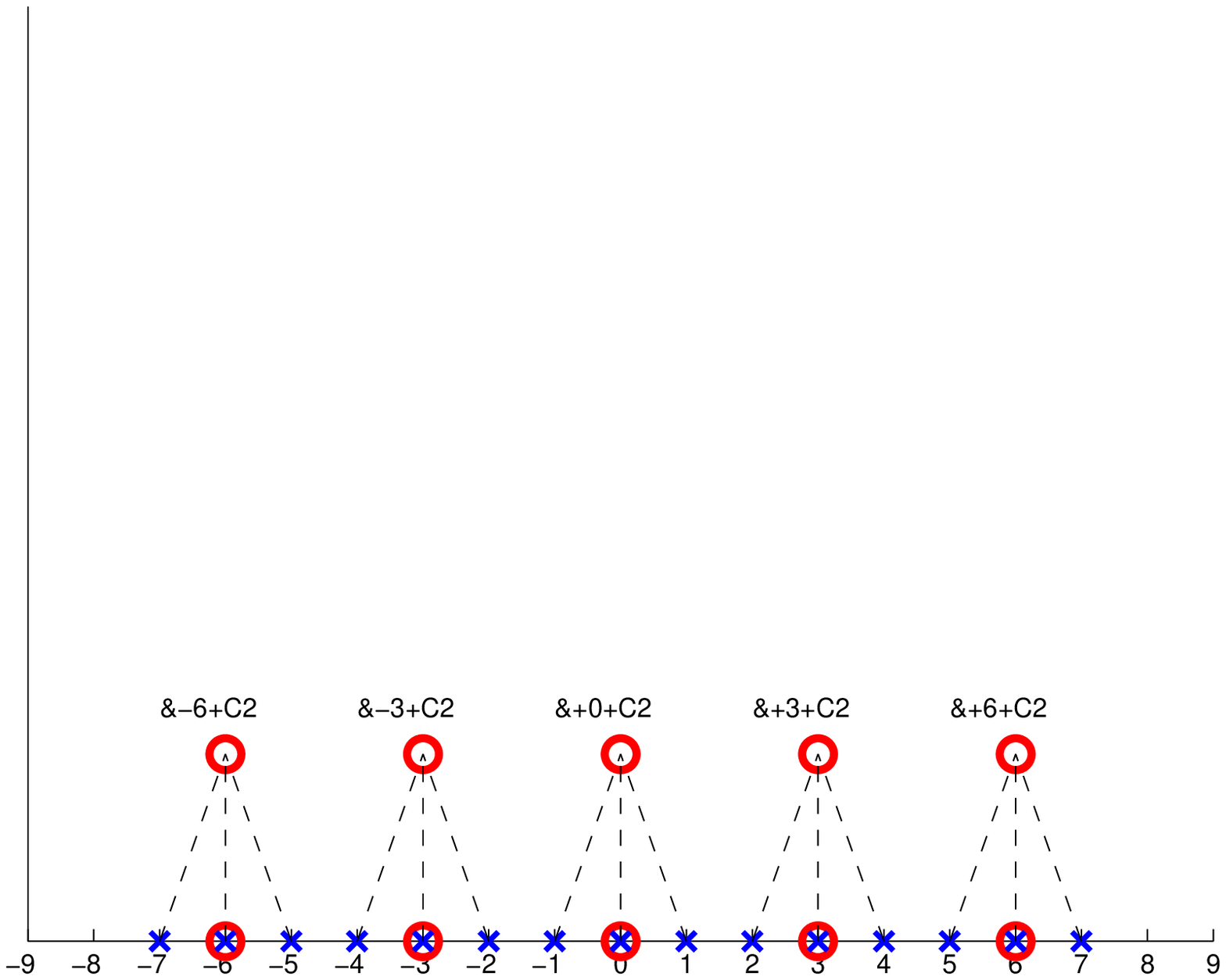}
}
\vspace{-3mm}
\caption{One-dimensional example: the circles are points of the codebook of the first user. The dots are the superposition of the the two codebooks ($L_1=3, L_0=15$).}
\vspace{-6mm}
\label{fig:one_dim_example}
\end{figure}
Thus, from the point of view of the decoder, the sum-constellation may have resulted from transmission of a PAM codebook of a single user. Indeed, it is a PAM constellation with $\cardinality{\mset{C}_1} \cdot \cardinality{\mset{C}_2}$ points, where the distance between the points is equal to that of $\mset{C}_2$. The error probability can be calculated directly from here. 
%
%
This demonstrates another effect of the power constraint that did not appear in the modulo-lattice MAC channel:
The Minkowski addition of the codebooks can be seen as \emph{tiling} the codebook with the weaker power, $\mset{C}_2$, around the points of the codebook with the stronger power, $\mset{C}_1$, such that the resulting points are equidistant. For such a construction, the powers of the codebooks must be \emph{non-equal}, and in particular, the power of the weaker user should be such that its PAM constellation fits between two points of the PAM constellation of the stronger user (see Figure~\ref{fig:one_dim_example}). 


Note that if the constellations were designed independently, e.g. by using PAM constellations with distances that do not result in a PAM constellation, some points would have been closer, resulting in higher error probability. 
In the high-dimensional case we strive to preserve this distance uniformity behavior. We will thus require that the decoding (Voronoi) region of the corresponding single-user codebook, as seen by the decoder, remains the same as the Voronoi region of $\mset{C}_2$.

Extending this example to higher dimensions will preserve this distance property, but will result in a shaping loss. Therefore, in addition to the distance properties between the codebook points, we want the corresponding codebook of the associated single-channel to have ``good'' shaping. Similarly to the coding technique which is described here, this will be achieved by using a triplet of ``good'' nested lattices.
In the one-dimensional case of PAM constellation, the tiling was perfect, in the sense that with an appropriate power-pair and an appropriate rate-pair, the resulting codebook as seen at the decoder is a PAM constellation with the sum powers and sum rates.
However, the coding technique which is proposed in the sequel for the multidimensional case does not accomplish a perfect shaping region as the one of the associated single-user channel, and therefore does not achieve the single-user error exponent. Furthermore, there is a capacity loss, which explains some of the loss in Theorem~\ref{thm:Gaussian} in respect with the capacity region of the MAC channel (the other loss which disappears at high SNR, is due to using a suboptimal decoder).

\subsection{Codebooks Construction}
\label{sec:codebook_construction}


Consider the Gaussian MAC channel~\eqref{eq:awgn_mac:channel}
with SNRs $A_1,A_2$ and rates $R_1,R_2$. Recall that without loss of generality we assume that $A_1 \geq A_2$ (equivalently, $P_1 \geq P_2$). We only consider rate-pairs inside $\Rstruct$ \eqref{eq:rate_region}. Furthermore we assume that $A_2\geq 1$ ($P_2 \geq N$), otherwise this rate-region is empty.

In our construction, the actual transmission power of the second user may be lower than the constraint $P_2$, in order to increase the exponent when larger rates for the first user are sought. This actual power is given by: 
\begin{align}
\tilde{P}_2 \mDefine \min\left\{ P_2, P_1 \exp(-2 R_1) \right\}.
\end{align}
In addition, denote
\begin{align}
\tilde{N} \mDefine \tilde{P}_2 \exp(-2 R_2).
\end{align}
Notice that these choices ensure that $N \leq \tilde{N} \leq \tilde{P}_2 \leq P_2$.
%
%
The codebook generation uses a triplet of nested lattices: $\Lambda_0^{(n)} \subseteq \Lambda_1^{(n)} \subseteq \Lambda_2^{(n)}$, where the covering radii are given by:
\begin{align*}
r^\textrm{cov}_{\mathcal{V}_0^{(n)}} = \sqrt{n P_1}, \mspcd
r^\textrm{cov}_{\mathcal{V}_1^{(n)}} = \sqrt{n \tilde{P}_2}, \mspcd
r^\textrm{cov}_{\mathcal{V}_2^{(n)}} = \sqrt{n \tilde{N}},
\end{align*}
where $\mathcal{V}_i^{(n)}$ is the Voronoi region of $\Lambda_i^{(n)}$ w.r.t. the origin.
Each lattice sequence is both Rogers-good and Poltyrev-good.\footnote{Such a chain of nested lattices exists by~\cite{KrithivasanPradhan:2007:online}. Less restrictive constraints on the goodness of the lattices may suffice.}
Denote the volume of $\mathcal{V}_i^{(n)}$ by $V_i^{(n)}$.
Since the lattices are Rogers-good, we have:
\begin{align*}
&V_0^{(n)} = [2 \pi e (P_1 - \delta_{1,n})]^{n/2},
\\&V_1^{(n)} = [2 \pi e (\tilde{P}_2 - \delta_{2,n})]^{n/2},
\\&V_2^{(n)} = [2 \pi e (\tilde{N} - \delta_{3,n})]^{n/2},
\end{align*}
where $\delta_{i,n} \mGoesTo 0^+$ as $n \mGoesTo \infty$.
See Figure~\ref{fig:nested} for an illustration of the nested lattices.

\begin{figure}[htb]
\centering
\includegraphics[trim = 20mm 15mm 20mm 10mm, clip, scale=0.6]{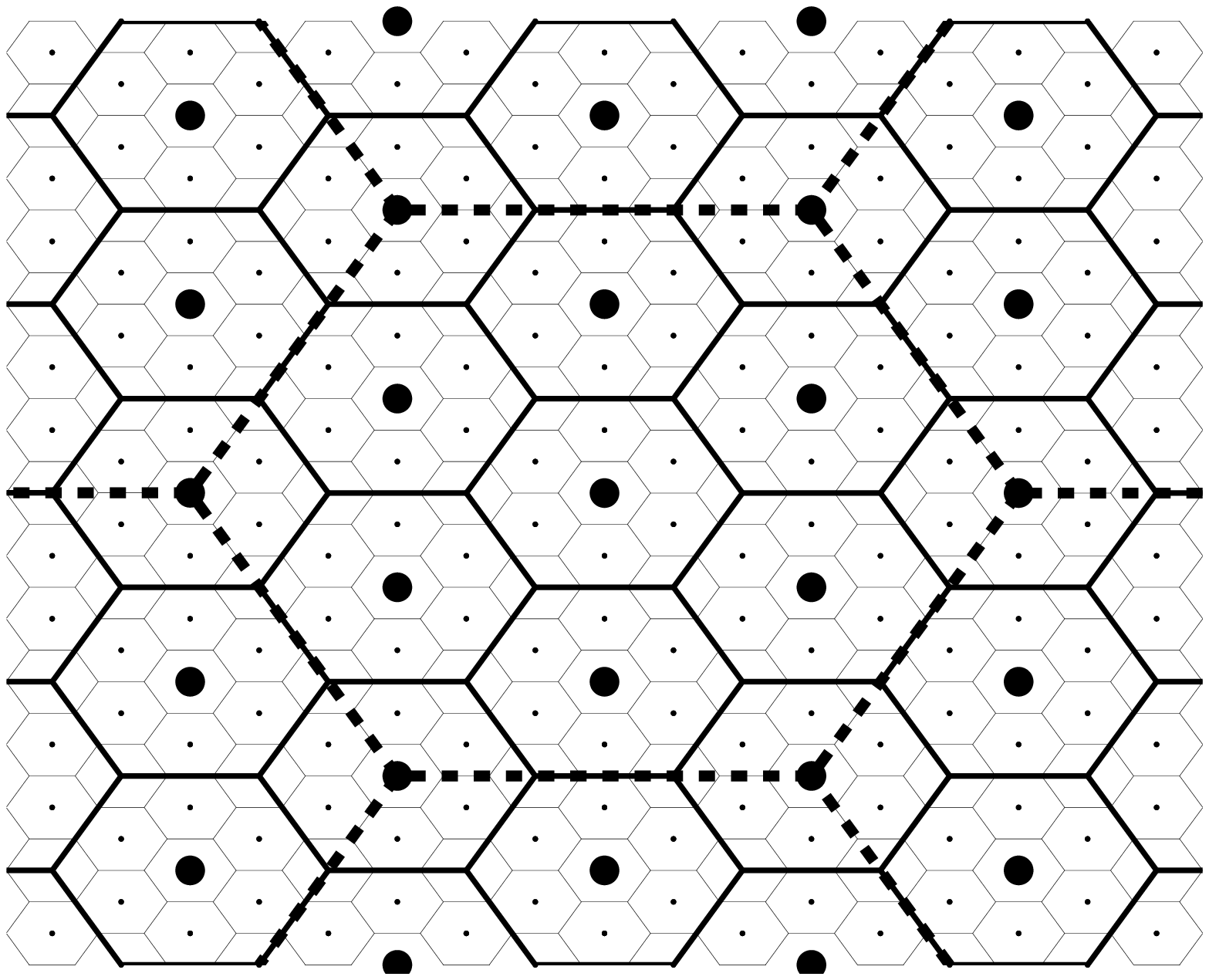}
\caption{Nested lattice with ratio 3,3: The thick-dashed line is the Voronoi partition of $\Lambda_0$, where the region in the center is $\mathcal{V}_0$, which is the shaping region of the first user.
The thick points are the points of $\Lambda_1$, and the thick-solid line is its Voronoi partition. The thick-solid line region in the center is $\mathcal{V}_1$, which is the shaping region of the second user.
The thin dots (together with the thick ones) are the points of $\Lambda_2$, and the thin-solid line is its Voronoi partition.}
\label{fig:nested}
\end{figure}


The codebook of the second user is given by the nested lattice code $( \Lambda_2^{(n)}, \Lambda_1^{(n)} )_\textrm{Vor}$, i.e.: $\mset{C}_2^{(n)} =  \Lambda_2^{(n)} \cap \mathcal{V}_1^{(n)}$.
The rate of the second user is given by the nesting ratio:
\begin{align}
\label{eq:R2}
R_2^{(n)} = \frac{1}{n} \log \frac{V_1}{V_2}
\mGoesToAs{n \mGoesTo \infty} \frac{1}{2} \log \frac{\tilde{P}_2}{\tilde{N}}
= R_2.
\end{align}
%
%
The codebook of the first user is given by the nested lattice code $( \Lambda_1^{(n)}, \Lambda_0^{(n)} )_\textrm{Vor}$, i.e.: $\mset{C}_1^{(n)} = \Lambda_1^{(n)} \cap \mathcal{V}_0^{(n)}$.
The rate of the first user results from the nesting ratio:
\begin{align}
\label{eq:R1}
R_1^{(n)} = \frac{1}{n} \log \frac{V_0}{V_1}
\mGoesToAs{n \mGoesTo \infty} \frac{1}{2} \log \frac{P_1}{\tilde{P}_2}
= R_1.
\end{align}
Since  $\Lambda_1^{(n)}$ and $\Lambda_2^{(n)}$ are Rogers-good lattices, as explained in Section~\ref{sec:good_lattices}, it follows that $\mset{C}_1^{(n)}$ and $\mset{C}_2^{(n)}$ satisfy the power constraints $P_1$ and $P_2$ respectively.
Notice that the Minkowski sum of the two codebooks is a subset of the fine lattice $\Lambda_2^{(n)}$.

\subsection{Error Probability Analysis}
\label{sec:error_analysis}

The Minkowski sum of the codebook pair, which is a subset of the fine lattice $\Lambda_2^{(n)}$, can be interpreted as the corresponding codebook of the associated single-user channel. Thus, we can use a \emph{lattice decoder} for joint decoding of the message pair. A lattice decoder is simply a lattice quantizer~\eqref{eq:quantizer}.
An achievable\footnote{We analyze the error probability of the nested lattice codebooks that are described in~\cite{KrithivasanPradhan:2007:online}. This ensemble builds on \emph{Construction~A}~\cite{Leoliger:1997}.} error exponent using lattice decoder is given by the Poltyrev exponent~\eqref{eq:poltyrev_exponent} with~$\mu$ given according to the Voronoi-to-noise ratio given by~\eqref{eq:voronoi_to_noise_ratio:n} (see~\cite{ErezLitsynZamir:2005}):
\begin{align*}
\mu^\textrm{struct} &= \mu^\textrm{struct}(R_1,R_2,A_1,A_2)
\\&\mDefine \lim_{n \mGoesTo \infty} \rho_{\Lambda^{(n)}}^2
\\&= \lim_{n \mGoesTo \infty} \frac{\left(V_2^{(n)}\right)^{2/n}}{2 \pi e N}
\\&= \frac{\tilde{N}}{N}
\\&= \frac{1}{N} \cdot \min\left\{ P_2, P_1 \exp(-2 R_1) \right\} \exp( -2 R_2 ) 
\\&= \min\left\{ A_2\exp( -2 R_2 ), A_1 \exp[-2 (R_1+R_2)] \right\}.
\\&= \min\left\{ \mu_2, \mu_1 \right\},
\end{align*}
where $\mu_1$ and $\mu_2$ are defined in \eqref{eq:mu1}, \eqref{eq:mu2}.
This proves Theorem~\ref{thm:Gaussian} for all rates in $\Rstruct$.
Since outside $\Rstruct$ this VNR satisfies $\mu^\textrm{struct} \leq 1$, it also proves the theorem for all rates. 

Note that the boundary where $\mu_1 = \mu_2$ corresponds for the case where $R_1 = \nicefrac{1}{2} \log \nicefrac{A_1}{A_2}$, and for smaller rates for the first user we use the maximal allowed power for transmission of the second user. In this case, $\mu_2$ is dominant.

\subsection{Performance Comparison}
\label{sec:comparison}

We now show that the suboptimal encoding-decoding scheme above outperforms the spherical-shells exponent for some rate pairs.

For given $A_1,A_2$ and a fixed $R_1=\nicefrac{1}{2} \log \nicefrac{A_1}{A_2}$, Figure~\ref{fig:comparison} compares the upper bound on the spherical-shells error-exponent $\ErUG$~\eqref{eq:spherical_shell:UbOnEr3} with the distributed-structure error exponent $\DistributedNestingErrexp$~\eqref{eq:Gaussian:MAC:structured_code:error_exponent}. We can see that below a certain rate $R_2$, distributed structure has a strictly larger error exponent than the spherical-shells one.
These exponents are also compared with $\EachievableSU(R_1+R_2,A_1+A_2)$, which is used here as a benchmark. Above the critical rate it is an upper bound.
Figure~\ref{fig:comparison:30:27} shows a high SNR case, where the second user has half of the power of the first user. The error exponent of the distributed structure code is strictly larger than the one of the spherical-shells codebooks in part of the rate region.
Figure~\ref{fig:comparison:50:25} shows a high SNR case, where the second user is much weaker than the first one; therefore it is almost a single user case.
Figure~\ref{fig:comparison:6:3} shows a low SNR case, where the second user looses since it is not equal to the single-user capacity $C(\tilde{A}_2)$.
Figure~\ref{fig:comparison:10:1} shows a case, where the second user is much weaker than the first one, and therefore the first user looses rate since it is not equal to the single-user capacity $C(A_1)$.
In the last two cases, the error exponent of the distributed structure code is lower than the one of the spherical-shells code.

\begin{figure*}[htb]
\centering
\subfigure[Balanced, high SNRs: $A_1 = 30$~dB, $A_2 = 27$~dB.]{
\psfrag{&alpha}[cc]{{\footnotesize $R_2/\frac{1}{2} \log A_2$}}
\includegraphics[scale=0.4]{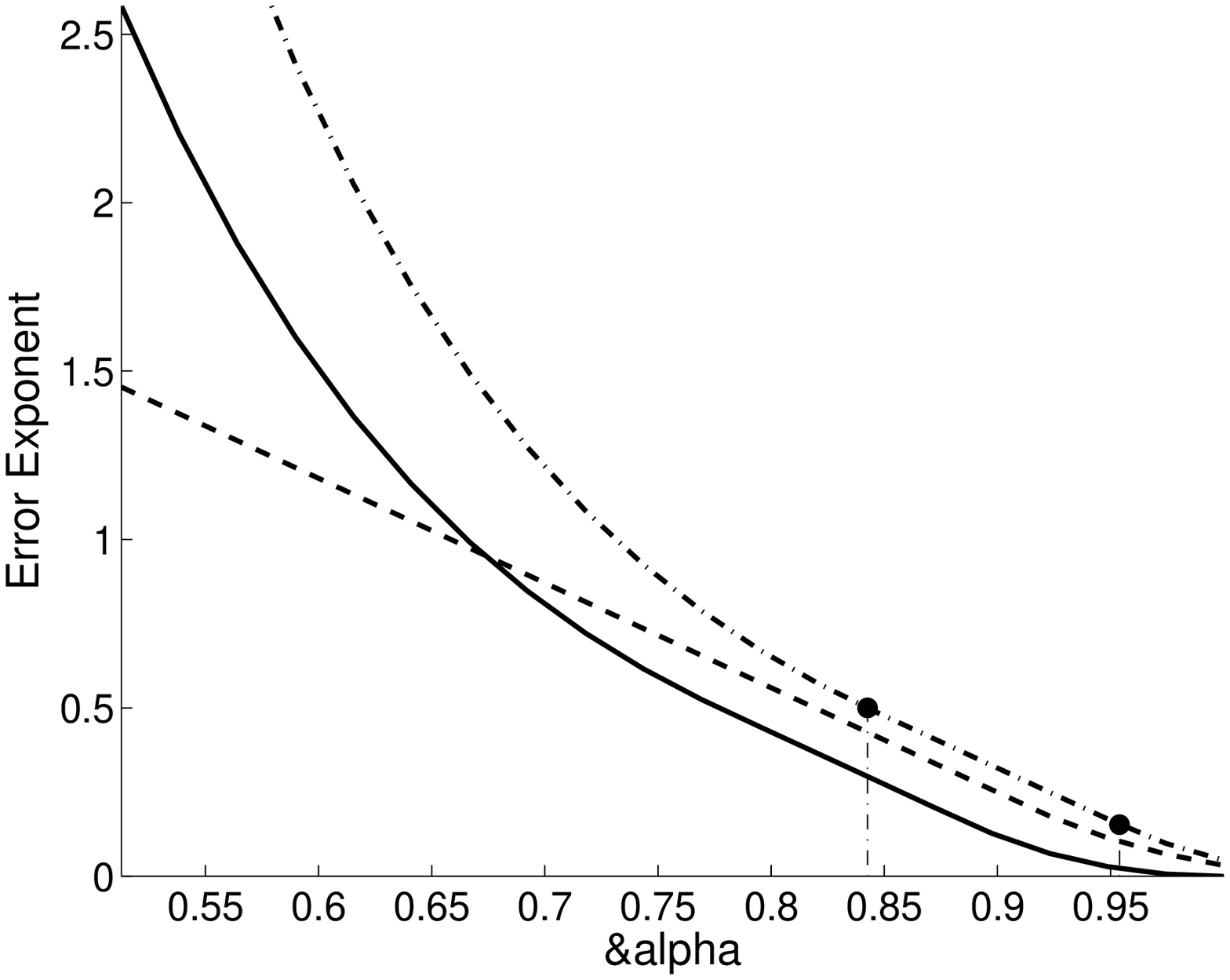}
\label{fig:comparison:30:27}
}
\subfigure[Unbalanced, high SNRs: $A_1 = 50$~dB, $A_2 = 25$~dB.]{
\psfrag{&alpha}[cc]{{\footnotesize $R_2/\frac{1}{2} \log A_2$}}
\includegraphics[scale=0.4]{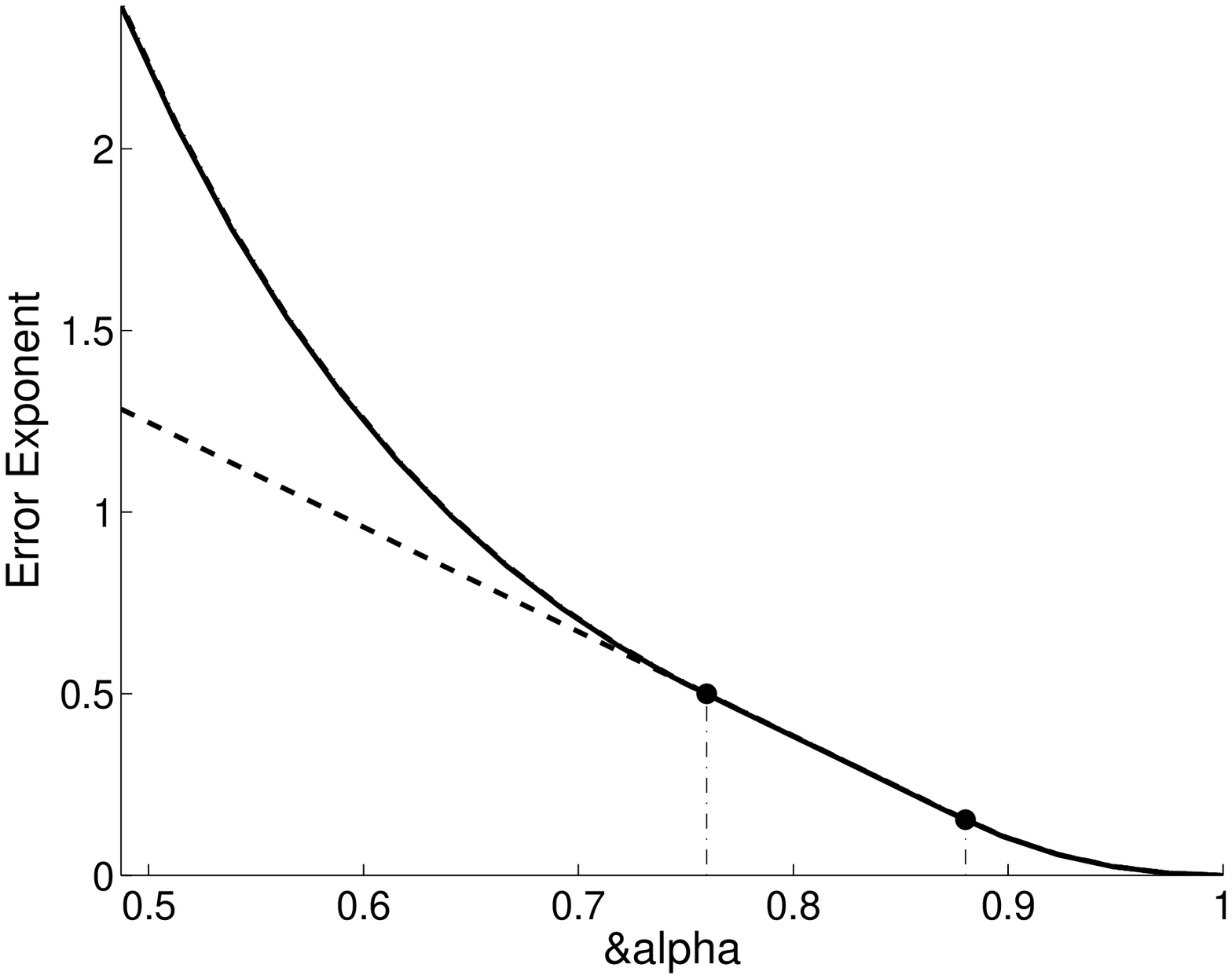}
\label{fig:comparison:50:25}
}
\subfigure[Balanced, low SNRs: $A_1 = 6$~dB, $A_2 = 3$~dB.]{
\psfrag{&alpha}[cc]{{\footnotesize $R_2/\frac{1}{2} \log A_2$}}
\includegraphics[scale=0.4]{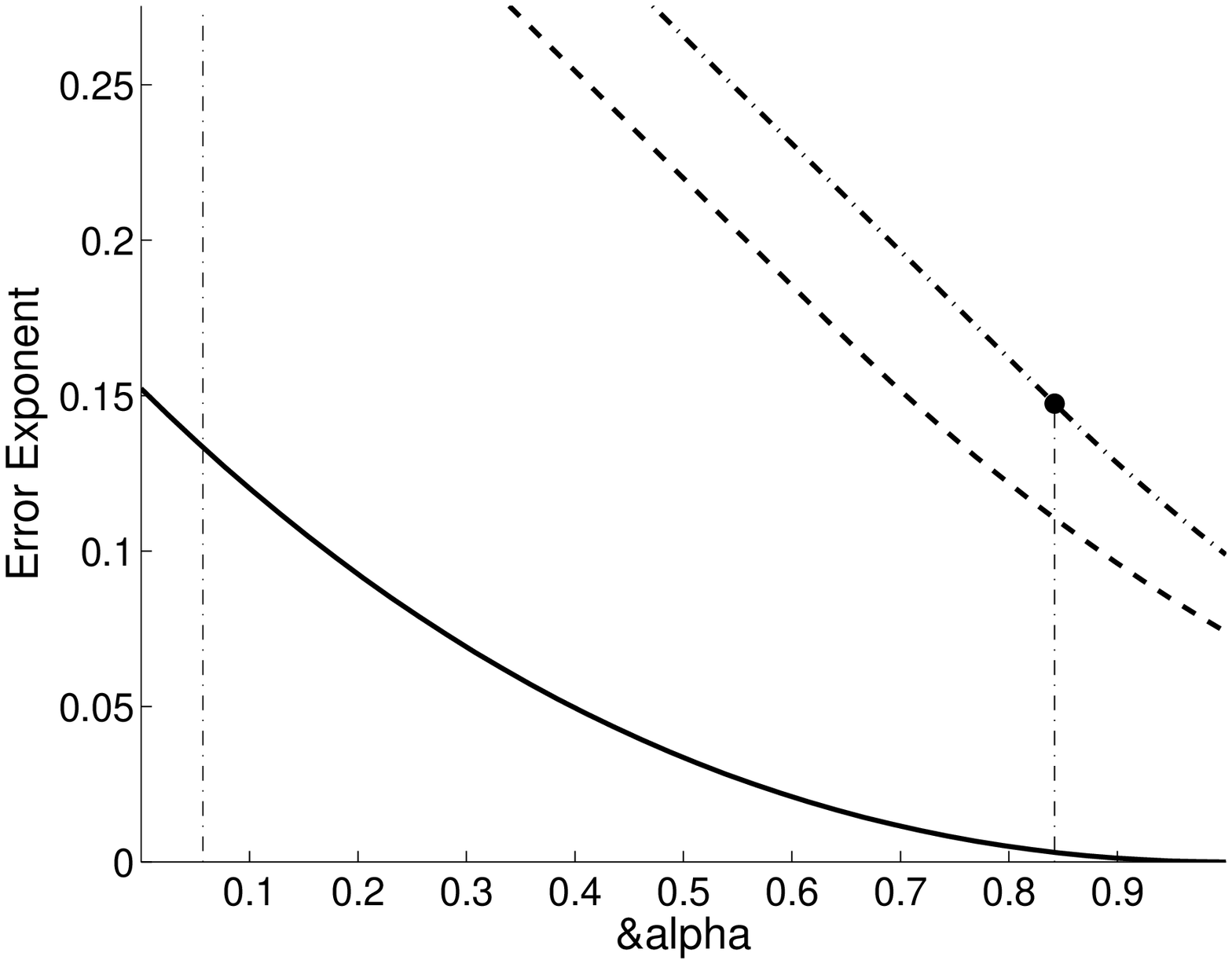}
\label{fig:comparison:6:3}
}
\subfigure[Unbalanced, SNRs: $A_1 = 10$~dB, $A_2 = 1$~dB.]{
\psfrag{&alpha}[cc]{{\footnotesize $R_2/\frac{1}{2} \log A_2$}}
\includegraphics[scale=0.4]{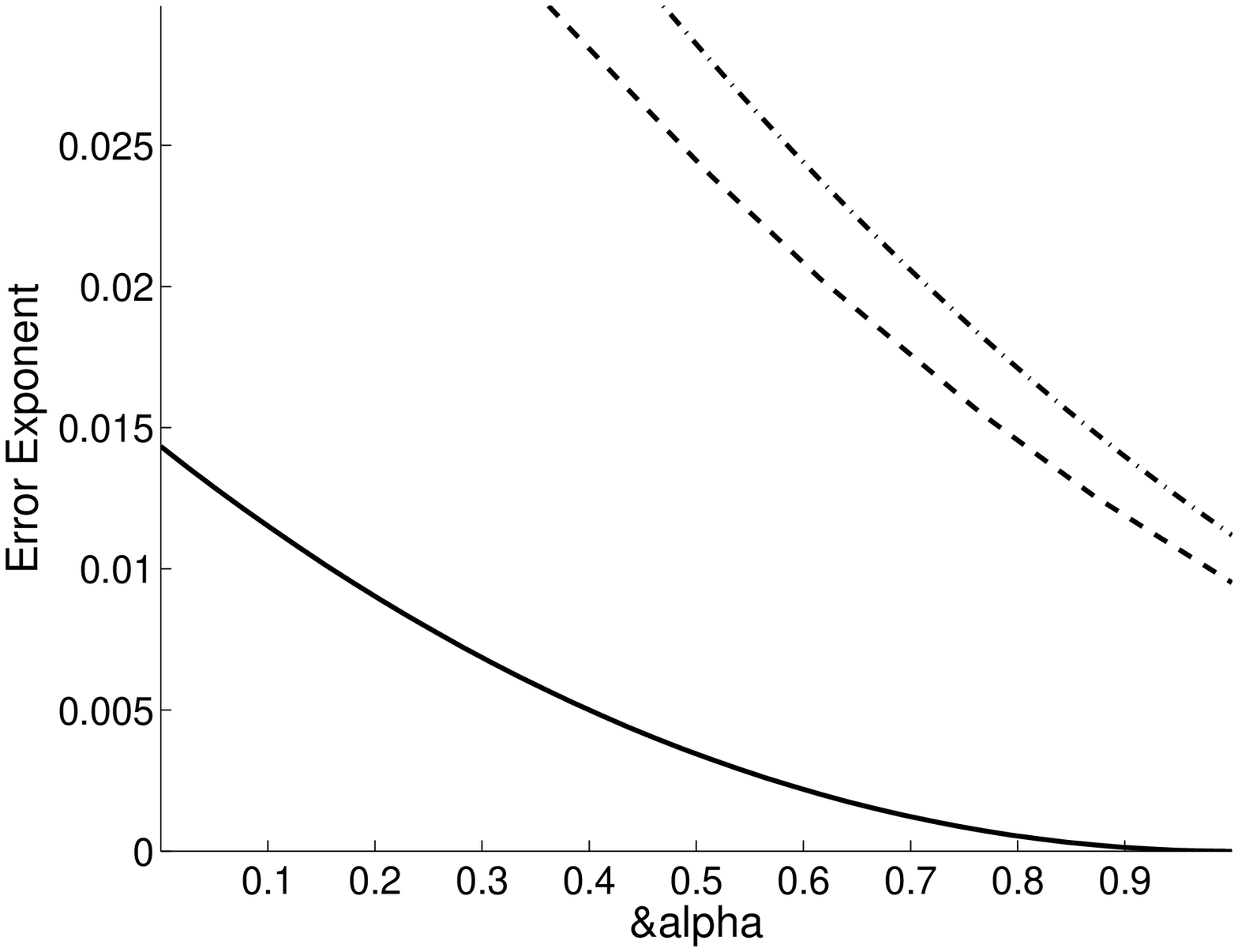}
\label{fig:comparison:10:1}
}
\caption{Comparing the error exponent of the spherical-shells code which is upper bounded by $\ErUG(R_1+R_2,A_1,A_2)$ (dashed line) to the error-exponent of the distributed structure $\DistributedNestingErrexp(R_1,R_2,A_1,A_2)$ (solid line). $\EachievableSU(R_1+R_2,A_1+A_2)$ of the associated single user channel is shown by the dash-dot line of the associated single user channel. The two dots indicate the expurgation rate $R_{ex}(A_1+A_2)$ and the critical rate $R_{cr}(A_1+A_2)$. The horizontal axis is the fraction of rate of the weak user from the maximum achievable by the distributed nesting technique, while $R_1=\nicefrac{1}{2} \log \nicefrac{A_1}{A_2}$ is fixed, and therefore $\mu_1=\mu_2$.}
\label{fig:comparison}
\end{figure*}

While for general (i.e., non-structured) infinite constellations, Poltyrev's exponent~\eqref{eq:poltyrev_exponent} in the range of squared Voronoi-to-noise ratio larger than $4$ is achieved by expurgation, for (infinite) lattices it is inherently achieved with high probability over the ensemble of~\cite{KrithivasanPradhan:2007:online}, since all codewords have the same error probability~\cite{Poltyrev:94:AWGN}.\footnote{For the ensemble of~\cite{KrithivasanPradhan:2007:online} the fraction of ``bad lattices'' goes to zero as $n \mGoesTo \infty$.}
We note that the distributed structure code is superior to the spherical-shells in (part of) the expurgation region of Poltyrev's exponent ($\mu \geq 4$).
This may imply that the ``inherent expurgation'' of lattices contributes to some of the gain of this code over the spherical-shells one.

%
\section{Discussion and Conclusion}
%
\label{sec:discussion}

By using linear codes, we have shown that for modulo-additive MAC channels with a prime alphabet size, the achievable error exponent is equal to the best known exponent of the associated single-user channel.
In addition, we have demonstrated that linear codes offer improvement to the best known MAC error exponent for ``almost additive'' channels. 

While we chose to present the results for two-user MAC channels, the approach immediately extends to any number of users, allowing for full expurgation of the combined linear code. It is therefore reasonable to expect that at low rates and at least for modulo-additive channels, the gain over the best previously known achievable error exponent will increase with the number of users. 
The approach of transforming a MAC channel into an additive one is also applicable to a wide variety of non-additive network setups, where structure is beneficial in terms of capacity.

In the Gaussian MAC case, the motivating observation of this approach is the fact that the sum of two uniformly distributed codebooks over spherical shells is not a uniform distribution over a spherical shell.
We use the fact that the sum of two nested lattices is a lattice, and this function is reversible. 
This technique, however, loses even in the sum rate. In order to understand the loss, we can view the codebook at the channel output as a tiling of the codebook of the weaker user over the codebook of the stronger one. We do not expect this tiling to be perfect, i.e., to constitute a good shaping region for the associated single-user channel, since the resulting shaping region exceeds the optimal shaping region of the associated single-user channel.
Even without perfect tiling, we expect that the results can be significantly improved.
Last, we note that a lattice version~\cite{ErezZamir:2008} of the transformation which was applied in the discrete case (Section~\ref{sec:discrete:transformation}) can be applied to almost additive Gaussian noise channels in order to improve the error exponent of the channel.

\bibliographystyle{IEEEtran.bst}
\bibliography{elih.bib}

\begin{thebibliography}{10}
\providecommand{\url}[1]{#1}
\csname url@samestyle\endcsname
\providecommand{\newblock}{\relax}
\providecommand{\bibinfo}[2]{#2}
\providecommand{\BIBentrySTDinterwordspacing}{\spaceskip=0pt\relax}
\providecommand{\BIBentryALTinterwordstretchfactor}{4}
\providecommand{\BIBentryALTinterwordspacing}{\spaceskip=\fontdimen2\font plus
\BIBentryALTinterwordstretchfactor\fontdimen3\font minus
  \fontdimen4\font\relax}
\providecommand{\BIBforeignlanguage}[2]{{%
\expandafter\ifx\csname l@#1\endcsname\relax
\typeout{** WARNING: IEEEtran.bst: No hyphenation pattern has been}%
\typeout{** loaded for the language `#1'. Using the pattern for}%
\typeout{** the default language instead.}%
\else
\language=\csname l@#1\endcsname
\fi
#2}}
\providecommand{\BIBdecl}{\relax}
\BIBdecl

\bibitem{SlepianWolf73MAC}
D.~Slepian and J.~K. Wolf, ``A coding theorem for multiple access channels with
  correlated sources,'' \emph{Bell System Tech. J.}, vol. vol. 52, pp.
  1037--1076, Sept. 1973.

\bibitem{Gallager:1985}
R.~G. Gallager, ``A perspective on multiaccess channels,'' \emph{IEEE Trans.
  Information Theory}, vol.~31, no.~2, pp. 124 -- 142, Mar. 1985.

\bibitem{PokornyWallmeier1985}
J.~Pokorny and H.~Wallmeier, ``Random coding bound and codes produced by
  permutations for the multiple-access channel,'' \emph{IEEE Trans. Information
  Theory}, vol.~31, no.~6, pp. 741 -- 750, Nov. 1985.

\bibitem{LiuHughes:1996}
Y.-S. Liu and B.~Hughes, ``A new universal random coding bound for the
  multiple-access channel,'' \emph{IEEE Trans. Information Theory}, vol.~42,
  no.~2, pp. 376 --386, Mar. 1996.

\bibitem{NazariAnastasopoulosPradhan:2010:arxiv}
A.~Nazari, A.~Anastasopoulos, and S.~S. Pradhan, ``Error exponent for
  multiple-access channels:lower bounds,'' \emph{CoRR}, vol. abs/1010.1303,
  2010.

\bibitem{KrithivasanPradhan:2009}
D.~Krithivasan and S.~Pradhan, ``Lattices for distributed source coding:
  Jointly {G}aussian sources and reconstruction of a linear function,''
  \emph{IEEE Trans. Information Theory}, vol.~55, no.~12, pp. 5628 --5651,
  2009.

\bibitem{WilsonNarayananPfisterSprintson:2010}
M.~P. Wilson, K.~R. Narayanan, H.~D. Pfister, and A.~Sprintson, ``Joint
  physical layer coding and network coding for bidirectional relaying,''
  \emph{IEEE Trans. Information Theory}, vol. IT-56, pp. 5641--5654, Nov. 2010.

\bibitem{PhilosofZamirErezKhisti:2011}
T.~Philosof, R.~Zamir, U.~Erez, and A.~Khisti, ``Lattice strategies for the
  dirty multiple access channel,'' \emph{IEEE Trans. Information Theory}, vol.
  IT-57, pp. 5006--5035, Aug. 2011.

\bibitem{ErezZamir:2008}
U.~Erez and R.~Zamir, ``A modulo-lattice transformation for multiple-access
  channels,'' in \emph{Electrical and Electronics Engineers in Israel, 2008.
  IEEEI 2008. IEEE 25th Convention of}, Dec. 2008, pp. 836 --840.

\bibitem{GallagerBook1968}
R.~G. Gallager, \emph{Information Theory and Reliable Communication}.\hskip 1em
  plus 0.5em minus 0.4em\relax John Wiley \& Sons, 1968.

\bibitem{WengPradhanAnastasopoulos:2008}
L.~Weng, S.~S. Pradhan, and A.~Anastasopoulos, ``Error exponent regions for
  {G}aussian broadcast and multiple-access channels,'' \emph{IEEE Trans.
  Information Theory}, vol. IT-54, pp. 2919--2942, Jul. 2008.

\bibitem{ErezZamir2001}
U.~Erez and R.~Zamir, ``Error exponents of modulo-additive noise channels with
  side information at the transmitter,'' \emph{IEEE Trans. Information Theory},
  vol. IT-47, pp. 210--218, Jan. 2001.

\bibitem{Dobrushin63}
R.~L. Dobrushin, ``Asymptotic optimality of group and systematic codes for some
  channels,'' \emph{Theor. Probab. Appl.}, vol.~8, pp. 52--66, 1963.

\bibitem{BargForney2002}
A.~Barg and J.~Forney, G.D., ``Random codes: minimum distances and error
  exponents,'' \emph{IEEE Trans. Information Theory}, vol.~48, no.~9, pp. 2568
  -- 2573, sep 2002.

\bibitem{ErezLitsynZamir:2005}
U.~Erez, S.~Litsyn, and R.~Zamir, ``Lattices which are good for (almost)
  everything,'' \emph{IEEE Trans. Information Theory}, vol. IT-51, pp.
  3401--3416, Oct. 2005.

\bibitem{Rogers59}
C.~A. Rogers, ``Lattice coverings of space,'' \emph{Mathematica}, vol.~6, pp.
  33--39, 1959.

\bibitem{Poltyrev:94:AWGN}
G.~Poltyrev, ``On coding without restrictions for the {AWGN} channel,''
  \emph{IEEE Trans. Information Theory}, vol. IT-40, pp. 409--417, March 1994.

\bibitem{KrithivasanPradhan:2007:online}
D.~Krithivasan and S.~S. Pradhan, ``A proof of the existence of good nested
  lattices,'' www.eecs.umich.edu/techreports/systems/cspl/cspl-384.pdf, 2007.

\bibitem{ErezZamir:2004}
U.~Erez and R.~Zamir, ``Achieving 1/2 log (1+{SNR}) on the {AWGN} channel with
  lattice encoding and decoding,'' \emph{IEEE Trans. Information Theory},
  vol.~50, no.~10, pp. 2293 -- 2314, Oct. 2004.

\bibitem{Leoliger:1997}
H.-A. Loeliger, ``Averaging bounds for lattices and linear codes,'' \emph{IEEE
  Trans. Information Theory}, vol.~43, no.~6, pp. 1767 --1773, Nov. 1997.

\end{thebibliography}

\end{document}